\documentclass[11pt,a4paper,reqno]{amsart}
\usepackage{accents}
\usepackage{acronym}
\usepackage{amsfonts}
\usepackage{amsmath}
\usepackage{amssymb}
\usepackage{amsthm}
\usepackage{arydshln}
\usepackage{braket}
\usepackage{cancel}
\usepackage{cite}
\usepackage{color}
\usepackage{eucal}
\usepackage{geometry}
\usepackage{graphicx}
\usepackage{hyperref}
\usepackage{mathtools}
\usepackage{setspace}
\usepackage{subcaption}
\usepackage{tikz}
\usepackage{todonotes}

\newgeometry{top=2cm,bottom=2cm,outer=2cm,inner=2cm}

\newcommand{\bse}{\begin{subequations}}
\newcommand{\ese}{\end{subequations}}

\newtheorem{theorem}{Theorem}[section]
\newtheorem{definition}[theorem]{Definition}
\newtheorem{lemma}[theorem]{Lemma}
\newtheorem{remark}[theorem]{Remark}

\newtheorem{proposition}[theorem]{Proposition}

\numberwithin{equation}{section}

\DeclareMathOperator{\tr}{tr}

\DeclareMathOperator*{\pf}{pf}

\def\undertilde#1{\mathord{\vtop{\ialign{##\crcr
$\hfil\displaystyle{#1}\hfil$\crcr\noalign{\kern1.5pt\nointerlineskip}
$\hfil\tilde{}\hfil$\crcr\noalign{\kern-6.5pt}}}}}
\def\underhat#1{\mathord{\vtop{\ialign{##\crcr
$\hfil\displaystyle{#1}\hfil$\crcr\noalign{\kern1.5pt\nointerlineskip}
$\hfil\hat{}\hfil$\crcr\noalign{\kern-6.5pt}}}}}
\def\underbar#1{\mathord{\vtop{\ialign{##\crcr
$\hfil\displaystyle{#1}\hfil$\crcr\noalign{\kern1.5pt\nointerlineskip}
$\hfil\bar{}\hfil$\crcr\noalign{\kern-6.5pt}}}}}
\def\undercheck#1{\mathord{\vtop{\ialign{##\crcr
$\hfil\displaystyle{#1}\hfil$\crcr\noalign{\kern1.5pt\nointerlineskip}
$\hfil\check{}\hfil$\crcr\noalign{\kern-6.5pt}}}}}

\newcommand{\rd}{\mathrm{d}}
\newcommand{\re}{\mathrm{e}}

\newcommand{\rD}{\mathrm{D}}

\newcommand{\cU}{\mathcal{U}}
\newcommand{\cV}{\mathcal{V}}
\newcommand{\cW}{\mathcal{W}}

\newcommand{\bc}{\boldsymbol{c}}
\newcommand{\tbc}{{}^{t\!}\boldsymbol{c}}
\newcommand{\bu}{\boldsymbol{u}}

\newcommand{\ba}{\boldsymbol{a}}
\newcommand{\tba}{{}^{t\!}\boldsymbol{a}}

\newcommand{\tbb}{{}^{t\!}\boldsymbol{b}}
\newcommand{\be}{\boldsymbol{e}}
\newcommand{\tbe}{{}^{t\!}\boldsymbol{e}}
\newcommand{\bk}{\boldsymbol{k}}
\newcommand{\tbk}{{}^{t\!}\boldsymbol{k}}

\newcommand{\bx}{\boldsymbol{x}}

\newcommand{\bA}{\boldsymbol{A}}

\newcommand{\bB}{\boldsymbol{B}}
\newcommand{\bC}{\boldsymbol{C}}
\newcommand{\tbC}{{}^{t\!}\boldsymbol{C}}

\newcommand{\bI}{\boldsymbol{I}}

\newcommand{\bL}{\boldsymbol{L}}
\newcommand{\tbL}{{}^{t\!}\boldsymbol{L}}
\newcommand{\bM}{\boldsymbol{M}}
\newcommand{\bO}{\boldsymbol{O}}

\newcommand{\bP}{\boldsymbol{P}}
\newcommand{\bQ}{\boldsymbol{Q}}
\newcommand{\bR}{\boldsymbol{R}}

\newcommand{\bU}{\boldsymbol{U}}
\newcommand{\tbU}{{}^{t\!}\boldsymbol{U}}

\newcommand{\bLd}{\mathbf{\Lambda}}
\newcommand{\tbLd}{{}^{t\!}\boldsymbol{\Lambda}}
\newcommand{\bOg}{\boldsymbol{\Omega}}
\newcommand{\tbOg}{{}^{t\!}\boldsymbol{\Omega}}

\newcommand{\bphi}{\boldsymbol{\phi}}

\newcommand{\kp}{\kappa}
\newcommand{\ld}{\lambda}
\newcommand{\og}{\omega}
\newcommand{\sg}{\sigma}

\newcommand{\Ld}{\Lambda}
\newcommand{\tLd}{{}^{t\!}\Lambda}
\newcommand{\Og}{\Omega}
\newcommand{\bblu}{\begin{color}{blue}}
\newcommand{\bred}{\begin{color}{red}}
\newcommand{\ecl}{\end{color}}

\acrodef{1D}[1D]{one-dimensional}
\acrodef{2D}[2D]{two-dimensional}
\acrodef{2DTL}[2DTL]{two-dimensional Toda lattice}
\acrodef{3D}[3D]{three-dimensional}
\acrodef{ABS}[ABS]{Adler--Bobenko--Suris}
\acrodef{bDT}[bDT]{binary Darboux transform}
\acrodef{BT}[BT]{B\"acklund transform}
\acrodef{BSQ}[BSQ]{Boussinesq}
\acrodef{CAC}[CAC]{consistency-around-the-cube}
\acrodef{DT}[DT]{Darboux transform}
\acrodef{DL}[DL]{direct linearisation}
\acrodef{DLT}[DLT]{direct linearising transform}
\acrodef{DS}[DS]{Drinfel'd--Sokolov}
\acrodef{DDeltaE}[D$\Delta$E]{differential-difference equation}
\acrodef{FG}[FG]{Fordy--Gibbons}
\acrodef{FX}[FX]{Fordy--Xenitidis}
\acrodef{GD}[GD]{Gel'fand--Dikii}
\acrodef{HM}[HM]{Hirota--Miwa}
\acrodef{HS}[HS]{Hirota--Satsuma}
\acrodef{KP}[KP]{Kadomtsev--Petviashvili}
\acrodef{KK}[KK]{Kaup--Kupershmidt}
\acrodef{KN}[KN]{Krichever--Novikov}
\acrodef{KdV}[KdV]{Korteweg--de Vries}
\acrodef{LL}[LL]{Landau--Lifshitz}
\acrodef{MDC}[MDC]{multi-dimensional consistency}
\acrodef{NLS}[NLS]{nonlinear Schr\"odinger}
\acrodef{NQC}[NQC]{Nijhoff--Quispel--Capel}
\acrodef{ODE}[ODE]{ordinary differential equation}
\acrodef{ODeltaE}[O$\Delta$E]{ordinary difference equation}
\acrodef{PDE}[PDE]{partial differential equation}
\acrodef{PDeltaE}[P$\Delta$E]{partial difference equation}
\acrodef{RHP}[RHP]{Riemann--Hilbert problem}
\acrodef{SK}[SK]{Sawada--Kotera}
\acrodef{sG}[sG]{sine--Gordon}
\acrodef{YB}[YB]{Yang--Baxter}

\title[On a coupled Kadomtsev--Petviashvili system associated with an elliptic curve]{On a coupled Kadomtsev--Petviashvili system \\ associated with an elliptic curve}
\author{Wei Fu}
\address[WF]{School of Mathematical Sciences and Shanghai Key Laboratory of Pure Mathematics and Mathematical Practice \\
East China Normal University \\ Shanghai 200241 \\ People's Republic of China}
\author{Frank W. Nijhoff}
\address[FWN]{School of Mathematics \\ University of Leeds \\ Leeds LS2 9JT \\ United Kingdom}
\begin{document}

\begin{abstract}
The coupled Kadomtsev--Petviashvili system associated with an elliptic curve,
proposed by Date, Jimbo and Miwa [J. Phys. Soc. Jpn., 52:766--771, 1983], is reinvestigated within the direct linearisation framework,
which provides us with more insights into the integrability of this elliptic model from the perspective of a general linear integral equation.
As a result, we successfully construct for the elliptic coupled Kadomtsev--Petviashvili system
not only a Lax pair composed of differential operators in $2\times2$ matrix form but also multi-soliton solutions with phases parametrised by points on the elliptic curve.
Dimensional reductions based on the direct linearisation, to the elliptic coupled Korteweg-de Vries and Boussinesq systems, are also discussed.
In addition, a novel class of solutions are obtained for the $D_\infty$-type Kadomtsev--Petviashvili equation with nonzero constant background as a byproduct.
\end{abstract}

\keywords{elliptic coupled KP, DKP, direct linearisation, dimensional reduction, Lax pair, $\tau$-function, soliton, nonzero constant background}

\maketitle

\section{Introduction}\label{S:Intro}

It is well-known that integrable systems often come in three different classes comprising rational, trigonometric/hyperbolic, and elliptic models.
Elliptic models are described by equations where there are parameters which are essentially moduli of elliptic curves,
or where the dependent variable appears in the argument of elliptic functions.
The recent history of the subject has taught us the remarkable finding that the theory of integrable systems is intimately linked to that of elliptic functions and curves.
One aspect of this connection is the fact that, as far as we know, the richest class of integrable systems are the ones associated with those curves.
For example, the Adler equation (i.e. the discrete Krichever--Novikov equation) \cite{Adl98} acts as the master equation among first-order partial difference equations; the elliptic Painlev\'e equation is on the top in Sakai's classification \cite{Sak01} of the discrete Painlev\'e equations.

Here we focus on a three-component integrable partial differential system given by
\bse\label{DJM}
\begin{align}
&\partial_3u=\frac{1}{4}\partial_1^3u+\frac{3}{2}(\partial_1u)^2+\frac{3}{4}\partial^{-1}_1\partial_2^2u+3g(1-vw), \\
&\partial_3v=-\frac{1}{2}\partial_1^3v-3(\partial_1u-3e)\partial_1v+\frac{3}{2}\partial_1\partial_2v+3(\partial_2u)v, \\
&\partial_3w=-\frac{1}{2}\partial_1^3w-3(\partial_1u-3e)\partial_1w-\frac{3}{2}\partial_1\partial_2w-3(\partial_2u)w,
\end{align}
\ese
where the solutions $u$, $v$ and $w$ are functions depending on variables $x_1$, $x_2$ and $x_3$.
The constant parameters $e$ and $g$ are moduli of an elliptic curve, given below in \eqref{Curve},
and $\partial_j$ denotes the partial-differential operator
\begin{align*}
\partial_j[\,\cdot\,]\doteq\frac{\partial}{\partial x_j}[\,\cdot\,],
\end{align*}
and $\partial_j^{-1}$ is a pseudo-differential operator which can be interpreted as an integral
\begin{align*}
\partial_{j}^{-1}[\,\cdot\,]\doteq\int^{x_j}[\,\cdot\,]\rd x_j
\end{align*}
in a standard way, see \cite{Dic03}.
Equation \eqref{DJM} is an alternative presentation of an elliptic \ac{KP}-type system
that was originally proposed by Date, Jimbo and Miwa, the form of which actually suggests a $(3+1)$-dimensional differential-difference system, see section 3 of \cite{DJM5}.
It was also pointed out by those authors that such a system follows from a similar construction of solutions
of the fully anisotropic \ac{LL} equation from a bilinear perspective as in \cite{DJKM83}.
Thus, one may infer that \eqref{DJM} in a sense is a \ac{KP} (higher-dimensional) analogue of the \ac{LL} equation.
Furthermore, from the form \eqref{DJM} of the elliptic \ac{KP} system we can conclude that
there is a connection (see section \ref{S:DKP} below) with the coupled \ac{KP} system
\bse\label{DKP}
\begin{align}
&\partial_t\cU=\frac{1}{4}\partial_x^3\cU+\frac{3}{2}\cU\partial_x\cU+\frac{3}{4}\partial^{-1}_x\partial_y^2\cU-6\partial_x(\cV\cW), \\
&\partial_t\cV=-\frac{1}{2}\partial_x^3\cV-\frac{3}{2}\cU\partial_x\cV+\frac{3}{2}\partial_x\partial_y\cV+\frac{3}{2}(\partial_x^{-1}\partial_y\cU)\cV, \\
&\partial_t\cW=-\frac{1}{2}\partial_x^3\cW-\frac{3}{2}\cU\partial_x\cW-\frac{3}{2}\partial_x\partial_y\cW-\frac{3}{2}(\partial_x^{-1}\partial_y\cU)\cW
\end{align}
\ese
proposed by Hirota and Ohta (cf. \cite{HO91} and also formula $(3.94)$ in \cite{Hir04}),
where the solutions $\cU$, $\cV$ and $\cW$ are functions of the independent variables $x$, $y$ and $t$,
with partial derivatives $\partial_x$, $\partial_y$ and $\partial_t$ respectively,
and where $\partial_x^{-1}$, defined in a similar way as above, is the pseudo-differential operator with respect to $x$.
Since \eqref{DKP} is one of the members in the DKP hierarchy (a \ac{KP}-type hierarchy associated with the infinite-dimensional Lie algebra $D_\infty$, see \cite{JM83})
which possesses a rich integrable structure in the theory of integrable systems,
we believe the understanding of elliptic model \eqref{DJM} will certainly yield additional insights into the integrability of many other nonlinear systems.
However, this remarkable elliptic integrable system has attracted little attention in the literature since the paper \cite{DJM5}, as far as we are aware.
For this reason, we believe that the elliptic coupled \ac{KP} system \eqref{DJM} deserves reinvestigation to explore further its integrability.

In the present paper, we shall adopt the \ac{DL} method to study the elliptic coupled \ac{KP} equation \eqref{DJM}.
Originally proposed by Fokas, Ablowitz and Santini for constructing a large class of solutions of nonlinear integrable partial differential equations,
the method which is based on formal singular linear integral equations \cite{FA81,FA83,SAF84},
was subsequently developed into a comprehensive framework to construct (discrete and continuous) integrable systems
and study their underlying algebraic structures, see e.g. \cite{NQLC83,Nij88},
and \cite{NQC83,QNCL84} for constructions of integrable discretisation of nonlinear partial differential equations,
and \cite{NCWQ84,NCW85,Nij85a,Nij85b,NC90} for the treatment of three-dimensional equations of \ac{KP}-type.
In \cite{Fu17a,Fu18b,Fu21a} the connection between the \ac{DL} and integrable systems on Lie algebras was developed.
A powerful tool emerging from \ac{DL}, first developed in \cite{NQLC83}, was an infinite matrix structure in the space of the spectral variable.
The associated infinite matrix representation of the linear integral equation allows us to turn the \ac{DL} into an algebraic method
which has been very effective in forging an understanding of the underlying integrability of the nonlinear integrable systems and their interconnections.
In \cite{NP03} and \cite{JN14}, the notion of elliptic infinite matrix (effectively amounting to an index-relabelling system of infinite matrices) was introduced.
This allows us to study integrable equations associated with elliptic curves within the \ac{DL} framework.

By reparametrising the time evolution and the Cauchy kernel for the \ac{LL} equation in the fermionic construction \cite{DJKM83}
and simultaneously considering a linear integral equation with a skew-symmetric integration measure
which was introduced in \cite{Fu17a,Fu18b,Fu21a} for the BKP-type equations,
we establish in this paper the \ac{DL} scheme for the elliptic coupled \ac{KP} system \eqref{DJM}.
This allows us to study the integrability of the elliptic coupled \ac{KP} system from a unified perspective.
As a result, we successfully derive the nonlinear equation \eqref{DJM} together with a suitable Lax pair from the elliptic infinite matrix structure.
Meanwhile, we construct the elliptic soliton solutions of \eqref{DJM} in terms of the $\tau$-function, which possesses a Pfaffian structure,
using a Pfaffian version of the well-known Laplace-type expansion formula (see \cite{Oka19} and also appendix \ref{S:Pfaff}),
as well as a Pfaffian analogue of the famous Frobenius formula for the determinants of elliptic Cauchy matrices (see \cite{DJKM83} and also appendix \ref{S:Cauchy}).
These results underpin not only the integrability of the elliptic coupled \ac{KP} system from the viewpoint of solvability,
but also induce a new class of solutions of the DKP equation with nonzero constant background as a byproduct.
In addition, we discuss dimensional reductions of \eqref{DJM},
from which we obtain the elliptic coupled \ac{KdV} and \ac{BSQ} systems together with their respective Lax pairs.

The paper is organised as follows.
In section \ref{S:InfMat}, we introduce the fundamental objects that will be used in construction of elliptic integrable systems,
including the notion of elliptic index-raising matrices and elliptic index labels.
The \ac{DL} scheme of the elliptic coupled \ac{KP} system is established in section \ref{S:DL} in the language of infinite matrices.
Section \ref{S:EllKP} is concerned with the construction of the elliptic coupled \ac{KP} system \eqref{DJM} and its Lax pair.
In the subsequent section \ref{S:Reduc}, we discuss dimensional reductions to the elliptic coupled \ac{KdV} and \ac{BSQ} systems.
The formulae of the elliptic soliton solutions to \eqref{DJM} are presented in section \ref{S:Sol}.
Finally, we explain in section \ref{S:DKP} how the soliton solutions to \eqref{DJM} generate those to the DKP equation with nonzero constant background.

\section{Infinite matrices and elliptic index labels}\label{S:InfMat}

We present an introduction to the fundamental objects that are needed in this paper, including the elliptic curve, infinite matrices, elliptic index-raising operators, etc.
These objects were introduced in \cite{NP03,JN14} for the DL construction of the so-called discrete and continuous elliptic \ac{KdV} and \ac{KP} equations.

The elliptic curve that we consider in this paper is of the form
\begin{align}\label{Curve}
k^2=K+3e+\frac{g}{K},
\end{align}
in which
\begin{align}\label{Moduli}
e\doteq e_1 \quad \hbox{and} \quad g\doteq(e_1-e_2)(e_1-e_3)
\end{align}
are the moduli of the curve, for $e_1$, $e_2$ and $e_3$ being the branch points of the standard Weierstrass elliptic curve $z^2=4(Z-e_1)(Z-e_2)(Z-e_3)$.
The elliptic curve \eqref{Curve} is parametrised by a uniformising variable $\kp$ through the coordinates
\begin{align}\label{Coordinates}
k=\frac{1}{2}\frac{\wp'(\kp)}{\wp(\kp)-e} \quad \hbox{and} \quad K=\wp(\kp)-e,
\end{align}
where $\wp$ and $\wp'$ denote the standard Weierstrass elliptic function and its first-order derivative, respectively.

We consider infinite matrices taking the form of $\bU=\left(U_{i,j}\right)_{\infty\times\infty}$
and infinite column and row vectors $\ba=(a_i)_{\infty\times 1}$ and $\tba=(a_i)_{1\times\infty}$.
We adopt the notations ${}^{t\!}(\cdot)$ for the transpose, $(\cdot)^{(i,j)}$ for the $(i,j)$-entry of an infinite matrix, and $(\cdot)^{(i)}$ for the $i$th-component of an infinite vector.
\begin{definition}
The index-raising infinite matrix $\Ld$ and its transpose $\tLd$ are defined by their respective $(i,j)$-entries
\begin{align}\label{Lambda}
\Ld^{(i,j)}\doteq\delta_{i+1,j} \quad \hbox{and} \quad \tLd^{(i,j)}\doteq\delta_{i,j+1}, \quad \forall i,j\in\mathbb{Z},
\end{align}
where $\delta_{\cdot,\cdot}$ is the standard Kronecker $\delta$-function defined as
\begin{align*}
\delta_{i,j}=
\left\{
\begin{array}{ll}
1, & i=j, \\
0, & i\neq j.
\end{array}
\right.
\end{align*}
\end{definition}
\begin{remark}
The infinite matrices $\Ld$ and $\tLd$ are entitled index-raising matrices because of the identities
\begin{align*}
(\Ld\bU)^{(i,j)}=\bU^{(i+1,j)} \quad \hbox{and} \quad (\bU\tLd)^{(i,j)}=\bU^{(i,j+1)}, \quad \forall i,j\in\mathbb{Z};
\end{align*}
in other words, the operation of $\Ld$ (resp. $\tLd$) from the left (resp. right) raises all the row (resp. column) indices of an infinite matrix by $1$.
Similarly, we have the identities
\begin{align*}
(\Ld\ba)^{(i)}=\ba^{(i+1)} \quad \hbox{and} \quad (\tba\tLd)^{(i)}=\tba^{(i+1)}, \quad \forall i\in\mathbb{Z},
\end{align*}
for infinite column and row vectors.
\end{remark}

\begin{definition}
The infinite projection matrix $\bO$ is defined by its $(i,j)$-entries
\begin{align*}
\bO^{(i,j)}\doteq\delta_{i,0}\delta_{0,j}, \quad \forall i,j\in\mathbb{Z}.
\end{align*}
\end{definition}
\begin{remark}
It is easily verified that the multiplication between $\bU$ and $\bO$ results in identities
\begin{align}\label{ProjectEntry}
(\bO\bU)^{(i,j)}=\delta_{i,0}\bU^{(0,j)} \quad \hbox{and} \quad (\bU\bO)^{(i,j)}=\bU^{(i,0)}\delta_{0,j}, \quad \forall i,j\in\mathbb{Z}.
\end{align}
This implies that the projection infinite matrix $\bO$ plays the role of mapping arbitrary $\bU$ to an infinite matrix of rank one.
Likewise, we have for arbitrary $\ba$ the identities
\begin{align}\label{ProjectComponent}
(\bO\ba)^{(i)}=\delta_{i,0}\ba^{(0)} \quad \hbox{and} \quad (\tba\bO)^{(i)}=\tba^{(0)}\delta_{0,i}, \quad \forall i\in\mathbb{Z}.
\end{align}
\end{remark}

\begin{remark}
Although \eqref{Lambda} suggests that the infinite matrices $\Ld$ and $\tLd$ are each other's inverse,
their role in the structure is such that they are never multiplied together,
and appear as distinct symbols separated by the projection matrix $\bO$ in the algebraic structure of infinite matrices.
\end{remark}

\begin{definition}
We define a special particular unit column vector $\be$ and its transpose $\tbe$ by their respective components
\begin{align*}
\be^{(i)}=\tbe^{(i)}=\delta_{i,0}, \quad \forall i\in\mathbb{Z}.
\end{align*}
\end{definition}

\begin{remark}
The infinite projection matrix $\bO$ can be written as the multiplication of $\be$ and $\tbe$, namely $\bO=\be\,\tbe$.
It is easily verified that $\be$ and $\tbe$ possess the properties
\begin{align}\label{UnitElement}
(\tbe\,\bU)^{(i)}=\bU^{(0,i)}, \quad (\bU\be)^{(i)}=\bU^{(i,0)} \quad \hbox{and} \quad \tbe\,\bU\be=\bU^{(0,0)}, \quad \forall i\in\mathbb{Z}
\end{align}
for an arbitrary infinite matrix $\bU$ as well as
\begin{align*}
\tbe\,\ba=\ba^{(0)} \quad \hbox{and} \quad \tba\,\be=\tba^{(0)}
\end{align*}
for arbitrary infinite column and row vectors $\ba$ and $\tba$.
\end{remark}

To deal with the elliptic integrable systems, we introduce the notion of elliptic index-raising matrices for future convenience.
\begin{definition}
The elliptic index-raising operators $\bLd$ and $\bL$ are defined as
\begin{align*}
\bLd\doteq\frac{1}{2}\frac{\wp'(\Ld)}{\wp(\Ld)-e} \quad {and} \quad \bL\doteq\wp(\Ld)-e,
\end{align*}
respectively; similarly, their respective transposes are defined by
\begin{align*}
\tbLd\doteq\frac{1}{2}\frac{\wp'(\tLd)}{\wp(\tLd)-e} \quad {and} \quad \tbL\doteq\wp(\tLd)-e.
\end{align*}
These elliptic index-raising operators should be understood as formal series expansions of $\Ld$ and $\tLd$, respectively.
\end{definition}

\begin{remark}
The elliptic index-raising operators $\bLd$, $\bL$ and $\tbLd$, $\tbL$ obey the elliptic curve relations
\begin{align}\label{EllIndex}
\bLd^2=\bL+3e+\frac{g}{\bL} \quad \hbox{and} \quad \tbLd^2=\tbL+3e+\frac{g}{\tbL},
\end{align}
respectively, as consequences of formulae \eqref{Curve} and \eqref{Coordinates}.
Here, for notational convenience when there is no issue of matrix ordering,
we have used the notations $\frac{1}{\bL}$ and $\frac{1}{\tbL}$ to denote the inverses of $\bL$ and $\tbL$, respectively.
\end{remark}

\begin{remark}
The formal operators $\bLd$, $\tbLd$, and $\bL$ and $\tbL$ should be understood in the following way.
Since $\bLd$ and $\bL$ commute, and similarly, $\tbLd$ and $\tbL$ commute,
we can consider a joint set of formal eigenvectors $\bc$ and $\tbc$ respectively, obeying the eigenvalue equations
\begin{align}\label{c}
\bLd\bc(\kp)=k\bc(\kp), \quad \bL\bc(\kp)=K\bc(\kp) \quad \hbox{and} \quad \tbc(\kp')\tbLd=k'\bc(\kp), \quad \tbc(\kp)\tbL=K'\bc(\kp),
\end{align}
where $(k,K)$ and $(k',K')$ are the points on the elliptic curve \eqref{Curve},
parametrised by their respective uniformising spectral parameters $\kp$ and $\kp'$ for $\kp,\kp'\in\mathbb{C}$, namely
we have relations
\begin{align*}
k=\frac{1}{2}\frac{\wp'(\kp)}{\wp(\kp)-e}, \quad K=\wp(\kp)-e, \quad k'=\frac{1}{2}\frac{\wp'(\kp')}{\wp(\kp')-e} \quad \hbox{and} \quad K'=\wp(\kp')-e.
\end{align*}
Thus, we can think of the points $(k,K)$ and $(k',K')$ on the curve as the `symbols' of these operators,
in some representation defined by the basis of infinite vectors $\bc(\kp)=(\kp^i)_{\infty\times 1}$ and $\tbc(\kp')=(\kp'^i)_{1\times\infty}$ respectively,
composed of the monomials of uniformising spectral parameters $\kp$ and $\kp'$.
Furthermore, these variables will play the role of the spectral parameters for the elliptic coupled \ac{KP} system.
\end{remark}

In the sections below, we shall only deal with the elliptic index-raising operations.
For this reason, we introduce elliptic index labels $[i,j]$ and $[i]$ for infinite matrices and infinite vectors, respectively,
compared with the above non-elliptic index labels $(i,j)$ and $(i)$.
\begin{definition}
In the elliptic index-labelling system,  the $[i,j]$-entries of the infinite matrix $\bU$ are defined by the following:
\bse\label{EllEntry}
\begin{align}
&\bU^{[2i,2j]}\doteq(\bL^i\bU\tbL^j)^{(0,0)}, \quad \bU^{[2i+1,2j+1]}\doteq(\bLd\bL^i\bU\tbL^j\tbLd)^{(0,0)}, \\
&\bU^{[2i+1,2j]}\doteq(\bLd\bL^i\bU\tbL^j)^{(0,0)}, \quad \bU^{[2i,2j+1]}\doteq(\bL^i\bU\tbL^j\tbLd)^{(0,0)},
\end{align}
\ese
for all $i,j\in\mathbb{Z}$. Likewise, we define the $[i]$th components of the infinite vectors $\ba$ and $\tba$ as follows:
\bse\label{EllComponent}
\begin{align}
&\ba^{[2i]}\doteq(\bL^i\ba)^{(0,0)}, \quad \ba^{[2i+1]}\doteq(\bLd\bL^i\ba)^{(0,0)}, \\
&\tba^{[2i]}\doteq(\tba\tbL^i)^{(0)}, \quad \tba^{[2i+1]}\doteq(\tba\tbL^i\tbLd)^{(0)},
\end{align}
\ese
for arbitrary $i\in\mathbb{Z}$.
\end{definition}
In the elliptic index-labelling system, $\bLd$ and $\tbLd$ (resp. $\bL$ and $\tbL$) play roles of order $1$ (resp. order $2$) index-raising operators.
We also comment that in concrete calculation, sometimes the curve relations in \eqref{EllIndex} are useful to reduce the powers of $\bLd$ and $\tbLd$.
For example, we have
\begin{align*}
(\bLd^2\bU)^{(0,0)}=\bU^{[2,0]}+3e\bU^{[0,0]}+g\bU^{[-2,0]} \quad \hbox{and} \quad (\bLd^3\bU)^{(0,0)}=\bU^{[3,0]}+3e\bU^{[1,0]}+g\bU^{[-1,0]},
\end{align*}
because of $\bLd^2\bU=\left(\bL+3e+g/\bL\right)\bU$ and $\bLd^3\bU=\bLd\bLd^2\bU=\bLd\left(\bL+3e+g/\bL\right)\bU$.

\section{Infinite matrix representation of the elliptic coupled KP system}\label{S:DL}

The essential ingredients in the construction of the elliptic coupled \ac{KP} system are the plane wave factor,
which defines the dynamics of the system in terms of the independent variables $x_j$, given by
\begin{align}\label{PWF}
\rho_n(\kp)\doteq\exp\left\{\sum_{j=0}^\infty k^{2j+1}x_{2j+1}+\sum_{j=1}^\infty\left[K^j-\left(\frac{g}{K}\right)^j\right]x_{2j}\right\}\left(\frac{K}{\sqrt{g}}\right)^n,
\end{align}
and the skew-symmetric Cauchy kernel
\begin{align}\label{Kernel}
\Og(\kp,\kp')\doteq\frac{K-K'}{k+k'}=\frac{k-k'}{1-\displaystyle\frac{g}{KK'}}.
\end{align}
The plane wave factor \eqref{PWF} and the kernel \eqref{Kernel} are reparametrisation of those objects first presented in \cite{DJKM83}
in the different context of a fermionic construction of the \ac{LL} equation.
\begin{definition}
For the given plane wave factor \eqref{PWF} and the Cauchy kernel \eqref{Kernel}, the linear integral equation of the elliptic coupled \ac{KP} system takes the form of
\begin{align}\label{Integral}
\bu_n(\kp)+\iint_{D}\rd\zeta(\ld,\ld')\rho_n(\kp)\Og(\kp,\ld')\rho_n(\ld')\bu_n(\ld)=\rho_n(\kp)\bc(\kp),
\end{align}
where the infinite column vector $\bu_n(\kp)$ is the wave function
whose components depend on the independent variables $x_j$ for $j\in\mathbb{Z}^+$ and the spectral parameter $\kp$,
and the integration measure $\rd\zeta$ and the integration domain $D$ must obey the antisymmetry property
\begin{align}\label{Measure}
\rd\zeta(\kp,\kp')=-\rd\zeta(\kp',\kp), \quad \forall (\kp,\kp')\in D.
\end{align}
\end{definition}
\begin{remark}
At the current stage, we do not specify the form of the integration measure.
The only requirement is that the associated homogeneous integral equation of \eqref{Integral} has only zero solution,
which guarantees the most general solution space from the perspective of the \ac{DL} approach.
\end{remark}
For the sake of construction of integrable systems, we need the infinite matrix $\bC_n$ defined as
\begin{align}\label{C}
\bC_n\doteq\iint_D\rd\zeta(\kp,\kp')\rho_n(\kp)\bc(\kp)\tbc(\kp')\rho_n(\kp'),
\end{align}
in which the integration measure and domain must satisfy the same antisymmetry property \eqref{Measure},
and also the infinite matrix $\bOg$ defined by
\begin{align}\label{Og}
\tbc(\kp')\bOg\bc(\kp)\doteq\Og(\kp,\kp').
\end{align}
From the definitions, it is reasonable to think of $\bC_n$ and $\bOg$ as the infinite matrix representation of the plane wave factor \eqref{PWF} and the kernel \eqref{Kernel}.
Due to the antisymmetry of the integration measure and the kernel,
it is verified that both $\bC_n$ and $\bOg$ are skew-symmetric, i.e. $\tbC_n=-\bC_n$ and $\tbOg=-\bOg$.
The key object towards nonlinear integrable systems in the direct linearisation is a potential matrix.
We give its definition as follows.
\begin{definition}
The potential matrix in the \ac{DL} is a double integral in terms of the spectral parameters, defined as
\begin{align}\label{Potential}
\bU_n\doteq\iint_D\rd\zeta(\kp,\kp')\bu_n(\kp)\tbc(\kp')\rho_n(\kp'),
\end{align}
in which $\bu_n(\kp)$ satisfies the linear integral equation \eqref{Integral},
and the measure is the same as the one for \eqref{Integral}, obeying the antisymmetry property \eqref{Measure}.
\end{definition}
\noindent
By following \eqref{Potential} and \eqref{Og}, we can reformulate the linear integral equation \eqref{Integral} as
\begin{align}\label{u}
\bu_n(\kp)=(1-\bU_n\bOg)\rho_n(\kp)\bc(\kp).
\end{align}
Performing the operation $\iint_D\rd\zeta(\kp,\kp')\eqref{u}\tbc(\kp')\rho_n(\kp')$, we can further derive
\begin{align}\label{U}
\bU_n=(1-\bU_n\bOg)\bC_n, \quad \hbox{or alternatively}, \quad \bU_n=\bC_n(1+\bOg\bC_n)^{-1}.
\end{align}
Equation \eqref{U} in some sense can be considered as the infinite matrix version of \eqref{u}.

We now introduce the $\tau$-function associated with the elliptic coupled \ac{KP} system.
\begin{definition}
The $\tau$-function is formally defined by
\begin{align}\label{tau}
\tau_n^2\doteq\det(1+\bOg\bC_n),
\end{align}
where $\bOg$ and $\bC_n$ are given by \eqref{Og} and \eqref{C}, respectively.
The determinant should be understood as the formal expansion
\begin{align*}
\det(1+\bOg\bC_n)=1+\sum_{i}(\bOg\bC_n)^{(i,i)}+\sum_{i<j}
\left|
\begin{matrix}
(\bOg\bC_n)^{(i,i)} & (\bOg\bC_n)^{(i,j)} \\
(\bOg\bC_n)^{(j,i)} & (\bOg\bC_n)^{(j,j)}
\end{matrix}
\right|+\cdots.
\end{align*}
We also remark that the determinant satisfies the identity $\ln[\det(1+\bOg\bC_n)]=\tr[\ln(1+\bOg\bC_n)]$.
\end{definition}

Our aim is to derive the dynamical relations of $\bU_n$ in terms of the continuous variables $x_j$ and the discrete variable $n$.
To realise this, we first investigate the evolutions of $\bC_n$.
Observing the identity \eqref{c} and the form of the plane wave factor \eqref{PWF}, we obtain the following dynamical relations:
\bse\label{CDyn}
\begin{align}
&\partial_{2j+1}\bC_n=\bLd^{2j+1}\bC_n+\bC_n\tbLd^{2j+1}, \label{CDyna} \\
&\partial_{2j}\bC_n=\left[\bL^j-\left(\frac{g}{\bL}\right)^j\right]\bC_n+\bC_n\left[\tbL^j-\left(\frac{g}{\tbL}\right)^j\right], \label{CDynb} \\
&\bC_{n+1}\frac{\sqrt{g}}{\tbL}=\frac{\bL}{\sqrt{g}}\bC_n, \label{CDync}
\end{align}
\ese
by differentiating $\bC_n$ with respect to $x_j$ and shifting $\bC_n$ with respect to $n$.
Next, from equations \eqref{Kernel} and \eqref{Og} we are able to derive
\begin{align}\label{OgAlg}
\bOg\bLd+\tbLd\bOg=\bO\bL-\tbL\bO \quad \hbox{and} \quad \bOg-g\,\tbL^{-1}\bOg\bL^{-1}=\bO\bLd-\tbLd\bO.
\end{align}
In fact, multiplying the two equations in \eqref{OgAlg} respectively by $\tbc(\kp')$ from left and $\bc(\kp)$ from the right yield
\begin{align*}
\Og(\kp,\kp')k+k'\Og(\kp,\kp')=K-K' \quad \hbox{and} \quad \Og(\kp,\kp')-\frac{g}{KK'}\Og(\kp,\kp')=k-k',
\end{align*}
namely \eqref{Kernel}; in other words, equations in \eqref{OgAlg} are nothing but the infinite matrix representation of the elliptic Cauchy kernel.
With the help of \eqref{OgAlg}, we can further derive the following relations for the infinite matrix $\bOg$ by mathematical induction:
\bse\label{OgDyn}
\begin{align}
&\bOg\bLd^{2j+1}+\tbLd^{2j+1}\bOg=\bO_{2j+1}\bL-\tbL\bO_{2j+1}, \label{OgDyna} \\
&\bOg\left[\bL^j-\left(\frac{g}{\bL}\right)^j\right]+\left[\tbL^j-\left(\frac{g}{\tbL}\right)^j\right]\bOg=\bO'_j\bLd-\tbLd\bO'_j, \label{OgDynb} \\
&\bOg\frac{\bL}{\sqrt{g}}-\frac{\sqrt{g}}{\tbL}\bOg=\bO_2\frac{\bL}{\sqrt{g}}, \label{OgDync}
\end{align}
\ese
in which
\begin{align*}
\bO_j\doteq\sum_{i=0}^{j-1}\left(-\tbLd\right)^i\bO\bLd^{j-1-i} \quad \hbox{and} \quad
\bO'_j\doteq\sum_{i=0}^{j-1}g^i\,\tbL^{-i}\bO\bL^{j-1-i}.
\end{align*}

\begin{proposition}\label{P:UDyn}
The infinite matrix $\bU_n$ defined by \eqref{Potential} satisfies the following continuous and discrete dynamical evolutions:
\bse\label{UDyn}
\begin{align}
&\partial_{2j+1}\bU_n=\bLd^{2j+1}\bU_n+\bU_n\tbLd^{2j+1}-\bU_n\left(\bO_{2j+1}\bL-\tbL\bO_{2j+1}\right)\bU_n, \label{UDyna} \\
&\partial_{2j}\bU_n=\left[\bL^j-\left(\frac{g}{\bL}\right)^j\right]\bU_n+\bU_n\left[\tbL^j-\left(\frac{g}{\tbL}\right)^j\right]-\bU_n\left(\bO'_j\bLd-\tbLd\bO'_j\right)\bU_n,  \label{UDynb} \\
&\bU_{n+1}\frac{\sqrt{g}}{\tbL}=\frac{\bL}{\sqrt{g}}\bU_n-\bU_{n+1}\bO_2\frac{\bL}{\sqrt{g}}\bU_n.  \label{UDync}
\end{align}
\ese
\end{proposition}
\begin{proof}
We only prove \eqref{UDynb} and \eqref{UDync}.
Differentiating \eqref{U} with respect to $x_{2j}$ gives rise to
\begin{align*}
\partial_{2j}\bU_n=(1-\bU_n\bOg)(\partial_{2j}\bC_n)-(\partial_{2j}\bU_n)\bOg\bC, \quad \hbox{namely} \quad
(\partial_{2j}\bU_n)(1+\bOg\bC_n)=(1-\bU_n\bOg)(\partial_{2j}\bC_n).
\end{align*}
Notice that the infinite matrix $\bC_n$ obeys the evolution given by \eqref{CDynb}. The above equation is reformulated as
\begin{align*}
(\partial_{2j}\bU_n)(1+\bOg\bC_n)=(1-\bU_n\bOg)\left[\bL^j-\left(\frac{g}{\bL}\right)^j\right]\bC_n+\bU_n\left[\tbL^j-\left(\frac{g}{\tbL}\right)^j\right].
\end{align*}
We can now replace $\bOg\left[\bL^j-\left(\frac{g}{\bL}\right)^j\right]$ by following \eqref{OgDynb}. This in turn implies
\begin{align*}
(\partial_{2j}\bU_n)(1+\bOg\bC_n)
=\left[\bL^j-\left(\frac{g}{\bL}\right)^j\right]\bC_n+\bU_n\left[\tbL^j-\left(\frac{g}{\tbL}\right)^j\right](1+\bOg\bC_n)-\bU_n\left(\bO'_j\bLd-\tbLd\bO'_j\right)\bC_n,
\end{align*}
which immediately results in \eqref{UDynb} by multiplying $(1+\bOg\bC_n)^{-1}$ from the right.
Equation \eqref{UDync} is derived by a similar approach. We shift \eqref{U} with respect to $n$ and obtain
\begin{align*}
\bU_{n+1}\frac{\sqrt{g}}{\tbL}=(1-\bU_{n+1}\bOg)\frac{\bL}{\sqrt{g}}\bC_n
\end{align*}
in virtue of \eqref{CDync}. By substituting $\bOg\frac{\bL}{\sqrt{g}}$ with the help of \eqref{OgDync}, this equation turns out to be
\begin{align*}
\bU_{n+1}\frac{\sqrt{g}}{\tbL}=\frac{\bL}{\sqrt{g}}\bC_n-\bU_{n+1}\left(\bO_2\frac{\bL}{\sqrt{g}}+\frac{\sqrt{g}}{\tbL}\bOg\right)\bC_n,
\end{align*}
i.e.
\begin{align*}
\bU_{n+1}\frac{\sqrt{g}}{\tbL}(1+\bOg\bC_n)=\frac{\bL}{\sqrt{g}}\bC_n-\bU_{n+1}\bO_2\frac{\bL}{\sqrt{g}}\bC_n.
\end{align*}
Multiplying $(1+\bOg\bC_n)^{-1}$ from the right, we end up with \eqref{UDync}.
Equation \eqref{UDyna} is proven in a similar way.
\end{proof}

\begin{proposition}\label{P:USym}
The infinite matrix $\bU_n$ satisfies the antisymmetry condition
\begin{align}\label{USym}
\tbU_n=-\bU_n, \quad \hbox{and consequently} \quad \bU_n^{[j,i]}=-\bU_n^{[i,j]}
\end{align}
in terms of the elliptic index labels $[i,j]$ for all $i,j\in\mathbb{Z}$.
\end{proposition}
\begin{proof}
Notice that $\bC_n$ and $\bOg$ are both skew-symmetric. We from \eqref{U} obtain
\begin{align*}
\tbU_n&={}^{t\!}\left[\bC_n(1+\bOg\bC_n)^{-1}\right]={}^{t\!}\left[\left(\bC_n^{-1}+\bOg\right)^{-1}\right] \\
&=\left(\tbC_n^{-1}+\tbOg\right)^{-1}=-\left(\bC_n^{-1}+\bOg\right)^{-1}=-\bC_n(1+\bOg\bC_n)^{-1}=-\bU_n,
\end{align*}
and subsequently $\bU_n^{[i,j]}=-\bU_n^{[j,i]}$ for all $i,j\in\mathbb{Z}$ by following the elliptic index labels defined by \eqref{EllEntry}.
\end{proof}

We can also follow the derivation of \eqref{UDyn} and construct the dynamical relations for the wave function $\bu_n(\kp)$.
\begin{proposition}\label{P:uDyn}
The wave function of the linear integral equation \eqref{Integral} obeys dynamical evolutions with respect to the continuous variables $x_j$ and the discrete variable $n$ as follows:
\bse\label{uDyn}
\begin{align}
&\partial_{2j+1}\bu_n(\kp)=\bLd^{2j+1}\bu_n(\kp)-\bU_n\left(\bO_{2j+1}\bL-\tbL\bO_{2j+1}\right)\bu_n(\kp), \label{uDyna} \\
&\partial_{2j}\bu_n(\kp)=\left[\bL^j-\left(\frac{g}{\bL}\right)^j\right]\bu_n(\kp)-\bU_n\left(\bO'_j\bLd-\tbLd\bO'_j\right)\bu_n(\kp), \label{uDynb} \\
&\bu_{n+1}(\kp)=\frac{\bL}{\sqrt{g}}\bu_n(\kp)-\bU_{n+1}\bO_2\frac{\bL}{\sqrt{g}}\bu_n(\kp). \label{uDync}
\end{align}
\ese
\end{proposition}
\begin{proof}
We only present the proof of \eqref{uDynb} and \eqref{uDync}. By differentiating \eqref{u} with respect to $x_{2j}$, we obtain
\begin{align*}
\partial_{2j}\bu_n(k)=(1-\bU_n\bOg)[\partial_{2j}\rho_n(\kp)]\bc(\kp)-(\partial_{2j}\bU_n)\bOg\rho_n(\kp)\bc(\kp).
\end{align*}
Equations \eqref{c} and \eqref{UDynb} can help us to reformulate the above equation as
\begin{align*}
\partial_{2j}\bu_n(k)={}&\left[\bL^j-\left(\frac{g}{\bL}\right)^j\right](1-\bU_n\bOg)\rho_n(\kp)\bc(\kp) \\
&-\bU_n\left\{\bOg\left[\bL^j-\left(\frac{g}{\bL}\right)^j\right]+\left[\tbL^j-\left(\frac{g}{\tbL}\right)^j\right]\bOg\right\}\rho(\kp)\bc(\kp) \\
&+\bU_n\left(\bO'_j\bLd-\tbLd\bO'_j\right)\bU_n\bOg\rho_n(\kp)\bc(\kp).
\end{align*}
Replacing $\bOg\left[\bL^j-\left(\frac{g}{\bL}\right)^j\right]+\left[\tbL^j-\left(\frac{g}{\tbL}\right)^j\right]\bOg$ by $\bO'_j\bLd-\tbLd\bO'_j$ according to \eqref{OgDynb}, we reach to
\begin{align*}
\partial_{2j}\bu_n(k)=\left[\bL^j-\left(\frac{g}{\bL}\right)^j\right](1-\bU_n\bOg)\rho_n(\kp)\bc(\kp)-\bU_n\left(\bO'_j\bLd-\tbLd\bO'_j\right)(1-\bU_n\bOg)\rho_n(\kp)\bc(\kp),
\end{align*}
which is nothing but equation \eqref{uDynb} according to \eqref{u}.
To derive \eqref{uDync}, we shift \eqref{u} with respect to $n$. This gives rise to
\begin{align*}
\bu_{n+1}(\kp)=(1-\bU_{n+1}\bOg)\rho_{n+1}(\kp)\bc(\kp)=(1-\bU_{n+1}\bOg)\frac{\bL}{\sqrt{g}}\rho_n(\kp)\bc(\kp).
\end{align*}
Then \eqref{OgDync} leads this equation to
\begin{align*}
\bu_{n+1}(\kp)=\frac{\bL}{\sqrt{g}}\rho_n(\kp)\bc(\kp)-\bU_{n+1}\left(\frac{\sqrt{g}}{\tbL}\bOg+\bO_2\frac{\bL}{\sqrt{g}}\right)\rho_n(\kp)\bc(\kp).
\end{align*}
Finally by substituting $\bU_{n+1}\frac{\sqrt{g}}{\tbL}$ with the help of \eqref{UDync}, the above equation turns out to be
\begin{align*}
\bu_{n+1}(\kp)=\frac{\bL}{\sqrt{g}}(1-\bU_n\bOg)\rho_n(\kp)\bc(\kp)-\bU_{n+1}\bO_2\frac{\bL}{\sqrt{g}}(1-\bU_n\bOg)\rho_n(\kp)\bc(\kp),
\end{align*}
namely equation \eqref{uDync} is proven in virtue of \eqref{u}. Equation \eqref{uDyna} can be proven through the same procedure.
\end{proof}
Finally, we present the dynamics of the $\tau$-function in terms of the indices of the infinite matrix $\bU_n$.
\begin{proposition}
The $\tau$-function satisfies dynamical evolutions
\bse\label{tauDyn}
\begin{align}\label{tauDyna}
2\partial_{2j+1}\ln\tau_n=\sum_{i=0}^{2j}(-1)^i\left(\bLd^{2j-i}\bL\bU_n\tbLd^i-\bLd^{2j-i}\bU_n\tbL\tbLd^i\right)^{(0,0)}
\end{align}
and
\begin{align}\label{tauDynb}
2\partial_{2j}\ln\tau_n=\sum_{i=0}^{j-1}g^i\left(\bL^{j-1-i}\bLd\bU_n\tbL^{-i}-\bL^{j-1-i}\bU_n\tbLd\tbL^{-i}\right)^{(0,0)}
\end{align}
with respect to the continuous arguments $x_j$, as well as
\begin{align}\label{tauDync}
\frac{\tau_{n+1}}{\tau_n}=1+g^{-1}\bU_n^{[3,2]} \quad \hbox{and} \quad \frac{\tau_{n-1}}{\tau_n}=1-\bU_n^{[1,0]}
\end{align}
with respect to the discrete argument $n$.
\ese
\end{proposition}
\begin{proof}
We first prove \eqref{tauDyna}. Differentiating the logarithm of the $\tau$-function with respective to $x_{2j+1}$ gives rise to
\begin{align*}
\partial_{2j+1}\ln\tau_n^2=\partial_{2j+1}\ln[\det(1+\bOg\bC_n)]=\partial_{2j+1}\tr[\ln(1+\bOg\bC_n)]=\tr[(1+\bOg\bC_n)^{-1}\bOg(\partial_{2j+1}\bC_n)].
\end{align*}
Replacing $\partial_{2j+1}\bC_n$ with the help of \eqref{CDyna} and \eqref{U}, we then obtain
\begin{align*}
\partial_{2j+1}\ln\tau_n^2=\tr[(1+\bOg\bC_n)^{-1}\bOg(\bLd^{2j+1}\bC_n+\bC_n\tbLd^{2j+1})]=\tr[(\bOg\bLd^{2j+1}+\tbLd^{2j+1}\bOg)\bU_n].
\end{align*}
Recall that $\bOg^{2j+1}\bLd+\tbLd^{2j+1}\bOg=\bO_{2j+1}\bL-\tbL\bO_{2j+1}$.
We end up with
\begin{align*}
2\partial_{2j+1}\ln\tau_n=\partial_{2j+1}\ln\tau_n^2=\tr[(\bO_{2j+1}\bL-\tbL\bO_{2j+1})\bU_n],
\end{align*}
which is nothing but \eqref{tauDyna}. Equation \eqref{tauDynb} follows from a similar derivation.
Next, performing the shift operation on \eqref{tau} we obtain
\begin{align*}
\tau_{n+1}^2&=\det(1+\bOg\bC_{n+1})=\det\left(1+\bOg\frac{\bL}{\sqrt{g}}\bC_n\frac{\tbL}{\sqrt{g}}\right)
=\det\left[1+\left(\bO_2\frac{\bL}{\sqrt{g}}+\frac{\sqrt{g}}{\tbL}\bOg\right)\bC_n\frac{\tbL}{\sqrt{g}}\right] \\
&=\det\left[1+\bOg\bC_n+(1+\bOg\bC_n)^{-1}\frac{\tbL}{\sqrt{g}}\bO_2\frac{\bL}{\sqrt{g}}\bC_n\right]
=\tau_n^2\det\left[1+(1+\bOg\bC_n)^{-1}\frac{\tbL}{\sqrt{g}}\bO_2\frac{\bL}{\sqrt{g}}\bC_n\right],
\end{align*}
where the second and third equalities hold because of \eqref{CDync} and \eqref{OgDync}, respectively.
Hence, this equation can further be rewritten as
\begin{align*}
\frac{\tau_{n+1}^2}{\tau_n^2}
&=\det\left[1+g^{-1}(1+\bOg\bC_n)^{-1}(\tbL,-\tbL\tbLd)\bO
\begin{pmatrix}
\bLd\bL \\
\bL
\end{pmatrix}
\bC_n\right] \\
&=\det\left[1+g^{-1}\left((1+\bOg\bC_n)^{-1}\tbL\be,-(1+\bOg\bC_n)^{-1}\tbL\tbLd\be\right)
\begin{pmatrix}
\tbe\bLd\bL\bC_n \\
\tbe\bL\bC_n
\end{pmatrix}
\right],
\end{align*}
where we have used the identity $\bO=\be\,\tbe$,
namely $\tau_{n+1}^2/\tau_n^2$ is of the form
\begin{align*}
\det\left[1+(\ba_1,\ba_2)
\begin{pmatrix}
\tbb_1 \\
\tbb_2
\end{pmatrix}
\right]
\end{align*}
for infinite column vectors $\ba_i$ and infinite row vectors $\tbb_i$ for $i=1,2$.
Using the rank $2$ Weinstein--Aronszajn formula
\begin{align*}
\det\left[1+(\ba_1,\ba_2)
\begin{pmatrix}
\tbb_1 \\
\tbb_2
\end{pmatrix}
\right]
=
\det\left[1+
\begin{pmatrix}
\tbb_1\ba_1 & \tbb_1\ba_2 \\
\tbb_2\ba_1 & \tbb_2\ba_2
\end{pmatrix}
\right]
\end{align*}
and also equation \eqref{U}, we obtain
\begin{align*}
\frac{\tau_{n+1}^2}{\tau_n^2}
&=\det
\begin{bmatrix}
1+g^{-1}\tbe\bLd\bL\bU_n\tbL\be & -g^{-1}\tbe\bLd\bL\bU_n\tbL\tbLd\be \\
g^{-1}\tbe\bL\bU_n\tbLd\be & 1-g^{-1}\tbe\bL\bU_n\tbL\tbLd\be
\end{bmatrix} \\
&=\det
\begin{bmatrix}
1+g^{-1}(\bLd\bL\bU_n\tbL)^{(0,0)} & -g^{-1}(\bLd\bL\bU_n\tbL\tbLd)^{(0,0)} \\
g^{-1}(\bL\bU_n\tbLd)^{(0,0)} & 1-g^{-1}(\bL\bU_n\tbL\tbLd)^{(0,0)}
\end{bmatrix} \\
&=\det
\begin{bmatrix}
1+g^{-1}\bU_n^{[3,2]} & -g^{-1}\bU_n^{[3,3]} \\
g^{-1}\bU_n^{[2,2]} & 1-g^{-1}\bU_n^{[2,3]}
\end{bmatrix}
=\left(1+g^{-1}\bU_n^{[3,2]}\right)^2,
\end{align*}
where the second equality holds because of the third equation in \eqref{UnitElement},
and we have made use of the antisymmetry property \eqref{USym} and \eqref{EllEntry} for the last equality.
Without loss of generality, we have $\tau_{n+1}/\tau_n=1+g^{-1}\bU_n^{[3,2]}$.
Likewise, by performing the backward shift operation on \eqref{tau}, we derive the other relation in \eqref{tauDync}.
\end{proof}

Equations \eqref{UDyna}, \eqref{UDynb} and \eqref{USym} together form the infinite matrix representation of the elliptic coupled \ac{KP} hierarchy,
describing possible dynamical evolutions and an algebraic constraint of the potential matrix.
Meanwhile, equations in \eqref{uDyn} form the infinite vector representation of the Lax pair of the elliptic coupled \ac{KP} hierarchy.
In the subsequent section, we shall present a closed-form multi-component nonlinear system composed of certain entries of the infinite matrix $\bU_n$
and a linear system based on a particular component of the infinite vector $\bu_n(\kp)$,
which respectively form the closed-form elliptic coupled \ac{KP} system and its Lax pair.

\section{Closed-form nonlinear system and associated linear problem}\label{S:EllKP}

We now construct the closed-form elliptic coupled \ac{KP} system and its Lax pair.
Our attention is only paid to the first nontrivial flow in the hierarchy,
namely the equation evolving with respect to the flow variables $x_1$, $x_2$ and $x_3$.
Our construction of the elliptic coupled \ac{KP} system is based on new variables as follows:
\bse\label{Component}
\begin{align}
&u_n\doteq\bU_n^{[2,0]}\equiv-\bU_n^{[0,2]}, \\
&v_n\doteq1+g^{-1}\bU_n^{[3,2]}\equiv1-g^{-1}\bU_n^{[2,3]}, \\
&w_n\doteq1-\bU_n^{[1,0]}\equiv1+\bU_n^{[0,1]}.
\end{align}
\ese
From the dynamical equations in \eqref{UDyn}, we are able to find a closed-form three-component system composed of the variables $u_n$, $v_n$ and $w_n$.
We present the result as the following theorem.
\begin{theorem}
Suppose that $\bu_n(\kp)$ is a solution to the linear integral equation \eqref{Integral} subject to \eqref{PWF}, \eqref{Kernel} and \eqref{Measure}.
The variables $u_n$, $v_n$ and $w_n$ defined by \eqref{Component} provide solutions to the elliptic coupled \ac{KP} system
\bse\label{EllKP:NL}
\begin{align}
&\partial_3u_n=\frac{1}{4}\partial_1^3u_n+\frac{3}{2}(\partial_1u_n)^2+\frac{3}{4}\partial^{-1}_1\partial_2^2u_n+3g(1-v_nw_n), \label{EllKP:NLa} \\
&\partial_3v_n=-\frac{1}{2}\partial_1^3v_n-3(\partial_1u_n-3e)\partial_1v_n+\frac{3}{2}\partial_1\partial_2v_n+3(\partial_2u_n)v_n, \label{EllKP:NLb} \\
&\partial_3w_n=-\frac{1}{2}\partial_1^3w_n-3(\partial_1u_n-3e)\partial_1w_n-\frac{3}{2}\partial_1\partial_2w_n-3(\partial_2u_n)w_n, \label{EllKP:NLc}
\end{align}
\ese
in which $e$ and $g$ are the moduli of the elliptic curve \eqref{Curve}.
Equation \eqref{EllKP:NL} is exactly the same as \eqref{DJM} by ignoring the discrete variable $n$.
\end{theorem}

\begin{proof}
Equations in \eqref{EllKP:NL} are verified by direct computation.
To realise this, we first of all need the fundamental formulae for the first-order derivatives $\partial_\mathfrak{j}\bU_n^{[i,j]}$
for $\mathfrak{j}\in\mathbb{Z}^+$ and $i,j\in\mathbb{Z}$.
Notice that
\begin{align*}
&\partial_{\mathfrak{j}}\bU_n^{[2i,2j]}=\partial_{\mathfrak{j}}(\bL^i\bU_n\tbL^j)^{(0,0)}=(\bL^i\partial_{\mathfrak{j}}\bU_n\tbL^j)^{(0,0)}, \\
&\partial_{\mathfrak{j}}\bU_n^{[2i+1,2j]}=\partial_{\mathfrak{j}}(\bLd\bL^i\bU_n\tbL^j)^{(0,0)}=(\bLd\bL^i\partial_{\mathfrak{j}}\bU_n\tbL^j)^{(0,0)}, \\
&\partial_{\mathfrak{j}}\bU_n^{[2i,2j+1]}=\partial_{\mathfrak{j}}(\bL^i\bU_n\tbL^j\tbLd)^{(0,0)}=(\bL^i\partial_{\mathfrak{j}}\bU_n\tbL^j\tbLd)^{(0,0)}, \\
&\partial_{\mathfrak{j}}\bU_n^{[2i+1,2j+1]}=\partial_{\mathfrak{j}}(\bLd\bL^i\bU_n\tbL^j\tbLd)^{(0,0)}=(\bLd\bL^i\partial_{\mathfrak{j}}\bU_n\tbL^j\tbLd)^{(0,0)}
\end{align*}
due to \eqref{EllEntry}.
Replacing all the $\partial_{\mathfrak{j}}\bU_n$ with the help of \eqref{UDyna} and \eqref{UDynb} for odd and even $\mathfrak{j}$, respectively,
we end up with the formulae for all the $\partial_\mathfrak{j}\bU_n^{[i,j]}$ expressed by various algebraic combinations of $\bU_n^{[i,j]}$.
For instance, when $i=2$ and $j=1$ the simplest formulae are given by
\begin{align*}
\partial_1\bU_n^{[2,0]}=\bU_n^{[3,0]}+\bU_n^{[2,1]}-\bU_n^{[2,0]}\bU_n^{[2,0]}+\bU_n^{[2,2]}\bU_n^{[0,0]}
\end{align*}
and
\begin{align*}
\partial_2\bU_n^{[2,0]}={}&\bU_n^{[4,0]}-g\bU_n^{[0,0]}+\bU_n^{[2,2]}-g\bU_n^{[2,-2]} \\
&-\bU_n^{[2,0]}\bU_n^{[3,0]}-\bU_n^{[2,2]}\bU_n^{[1,0]}+\bU_n^{[2,1]}\bU_n^{[2,0]}+\bU_n^{[2,3]}\bU_n^{[0,0]}.
\end{align*}
Then by iteration, we can further derive the general formulae for the higher-order derivatives of $\bU_n^{[i,j]}$ such as
$\partial_\mathfrak{j}^2\bU_n^{[i,j]}$, $\partial_\mathfrak{j}^3\bU_n^{[i,j]}$, $\partial_\mathfrak{j}^4\bU_n^{[i,j]}$,
$\partial_\mathfrak{i}\partial_\mathfrak{j}\bU_n^{[i,j]}$, etc. for $\mathfrak{i},\mathfrak{j}\in\mathbb{Z}^+$ and $i,j\in\mathbb{Z}$.
These formulae together with the antisymmetry property \eqref{USym} allow us to verify equations given by \eqref{EllKP:NL} in a purely algebraic way.
For example, in order to verify \eqref{EllKP:NLa} we need the corresponding formulae for
$\partial_1\bU_n^{[2,0]}$, $\partial_1^2\bU_n^{[2,0]}$, $\partial_1^4\bU_n^{[2,0]}$, $\partial_2^2\bU_n^{[2,0]}$ and $\partial_1\partial_3\bU_n^{[2,0]}$.
Then it is verified that
\begin{align*}
-\partial_1\partial_3\bU_n^{[2,0]}
+\frac{1}{4}\partial_1^4\bU_n^{[2,0]}+3\left(\partial_1\bU_n^{[2,0]}\right)\left(\partial_1^2\bU_n^{[2,0]}\right)+\frac{3}{4}\partial_2^2\bU_n^{[2,0]}
-3g\partial_1\left[\left(1+g^{-1}\bU_n^{[3,2]}\right)\left(1-\bU_n^{[1,0]}\right)\right]
\end{align*}
vanishes in virtue of $\bU_n^{[j,i]}=-\bU_n^{[i,j]}$.
This essentially means that
\begin{align*}
\partial_1\partial_3u_n=\partial_1\left(\frac{1}{4}\partial_1^3u_n+\frac{3}{2}(\partial_1u_n)^2\right)+\frac{3}{4}\partial_2^2u_n-3g\partial_1(v_nw_n)
\end{align*}
Observing the conditions on the asymptotic limits $u_n\rightarrow0$, $v_n\rightarrow1$ and $w_n\rightarrow1$ (which follow from \eqref{Component}),
we immediately derive \eqref{EllKP:NLa} by integration with respect to $x_1$.
Equations \eqref{EllKP:NLb} and \eqref{EllKP:NLc} are verified in the same manner.
\end{proof}

Next, we derive the bilinear form of the elliptic coupled \ac{KP} system.
Note that the simplest cases of \eqref{tauDyn} give us the bilinear transforms\footnote{
The last two equations imply that $v_n$ and $w_n$ are connected with each other through $v_nw_{n+1}=1$.
However, our aim is to present the elliptic coupled \ac{KP} system as a $(2+1)$-dimensional continuous integrable model,
where there is no dynamical evolution with respect to the discrete independent variable $n$.
Hence, $v_n$ and $w_n$ are treated as two separate variables, as we have done in \eqref{EllKP:NL}.
}
\begin{align}\label{uvwtau}
u_n=\partial_1\ln\tau_n, \quad v_n=\frac{\tau_{n+1}}{\tau_n} \quad \hbox{and} \quad w_n=\frac{\tau_{n-1}}{\tau_n}.
\end{align}
The transformations \eqref{uvwtau} help us to construct a closed-form bilinear system in terms of the $\tau$-function.
The result is presented as the theorem below.
\begin{theorem}
The $\tau$-function defined by \eqref{tau} satisfies the following system of bilinear equations:
\bse\label{EllKP:BL}
\begin{align}
&\left(\rD_1^4-4\rD_1\rD_3+3\rD_2^2\right)\tau_n\cdot\tau_n=24g\left(\tau_{n+1}\tau_{n-1}-\tau_n^2\right), \label{EllKP:BLa} \\
&\left(\rD_1^3+2\rD_3-3\rD_1\rD_2-18e\rD_1\right)\tau_{n+1}\cdot\tau_n=0, \label{EllKP:BLb} \\
&\left(\rD_1^3+2\rD_3+3\rD_1\rD_2-18e\rD_1\right)\tau_{n-1}\cdot\tau_n=0, \label{EllKP:BLc}
\end{align}
\ese
where $\rD_j$ stand for the bilinear derivatives (see appendix \ref{S:BL} for the definition) with respect to $x_j$.
\end{theorem}

\begin{proof}
These equations are obtained from \eqref{EllKP:NL} by using the transformations in \eqref{uvwtau} and identities \eqref{Log} and \eqref{BiLog}.
For instance, the bilinear transformations \eqref{uvwtau} reformulate \eqref{EllKP:NLa} as
\begin{align*}
\partial_1\partial_3\ln\tau_n=\frac{1}{4}\partial_1^4\ln\tau_n+\frac{3}{2}(\partial_1^2\ln\tau_n)^2+\frac{3}{4}\partial_2^2\ln\tau_n
+3g\left(1-\frac{\tau_{n+1}\tau_{n-1}}{\tau_n^2}\right)
\end{align*}
in terms of the $\tau$-function, which is exactly the same as \eqref{EllKP:BLa}, due to the logarithmic transformations give by \eqref{Log}.
Equations \eqref{EllKP:BLb} and \eqref{EllKP:BLc} are derived similarly from \eqref{EllKP:NLb} and \eqref{EllKP:NLc}, respectively,
where we need to make use of the bi-logarithmic transformations listed in \eqref{BiLog}.
We also comment that these bilinear equations can alternatively be verified directly by following the same procedure of deriving \eqref{EllKP:NL}.
This is because equations \eqref{tauDyn} establish the connection between $\tau_n$ and $\bU_n^{[i,j]}$.
\end{proof}

\begin{remark}
In fact, the last two equations in the bilinear equations \eqref{EllKP:BL} are equivalent to each other, which means either of them can be omitted.
However, we would still like to reserve both equations because here the main idea is to present the bilinear elliptic coupled \ac{KP} system
as a $(2+1)$-dimensional continuous closed-form system for $\tau_n$, $\tau_{n+1}$ and $\tau_{n-1}$ in terms of the independent variables $x_1$, $x_2$ and $x_3$,
rather than a $(3+1)$-dimensional differential-difference system for a single $\tau$-function in terms of the independent variables $x_1$, $x_2$, $x_3$ and $n$.
We also note that equations in \eqref{EllKP:BL} are reparametrisation of those bilinear equations in \cite{DJM5},
which, up to the moduli $e$ and $g$, have appeared in \cite{Leu04} (see also references therein) as the first few members in the bilinear DKP hierarchy.
\end{remark}

In addition, by introducing the scalar wave function $\phi_n\doteq\bu_n^{[0]}(\kp)$,
we can also construct from \eqref{uDyn} the associated linear problem for the nonlinear system \eqref{EllKP:NL}.
We conclude the result as the following theorem.
\begin{theorem}
For an arbitrary solution $\bu_n(\kp)$ to the linear integral equation \eqref{Integral},
the scalar wave function $\phi_n$ and the variables $u_n$, $v_n$ and $w_n$ defined by \eqref{Component}
satisfy the linear system for $\phi_n$ as follows\footnote{
These equations are effectively reparametrisation of the linear equations listed in \cite{DJM5}.
We also remark that a typo in the third equation in \cite{DJM5} has been fixed here.
}:
\bse\label{EllKP:L}
\begin{align}
&\partial_2\phi_n=-\left[\partial_1^2+2\partial_1u_n-3e\right]\phi_n+2\sqrt{g}v_nw_n\phi_{n+1}, \label{EllKP:La} \\
&\partial_3\phi_n=\left[\partial_1^3+3(\partial_1u_n)\partial_1+\frac{3}{2}\left(\partial_1^2u_n-\partial_2u_n\right)\right]\phi_n-3\sqrt{g}v_n(\partial_1 w_n)\phi_{n+1}, \label{EllKP:Lb} \\
&\sqrt{g}v_nw_n\phi_{n+1}+\sqrt{g}\phi_{n-1} \nonumber \\
&\qquad=\left[\partial_1^2-(\partial_1\ln w_n)\partial_1+2\partial_1u_n-3e+\frac{1}{2}\left(\partial_1^2\ln w_n+(\partial_1\ln w_n)^2+\partial_2\ln w_n\right)\right]\phi_n. \label{EllKP:Lc}
\end{align}
\ese
\end{theorem}

\begin{proof}
The idea of the proof is very similar to that of verifying \eqref{EllKP:NL}.
We start with deriving from \eqref{uDyn} the algebraic expressions of the derivatives $\partial_{\mathfrak{j}}\bu_n^{[i]}(\kp)$.
Following the elliptic index labels introduced by \eqref{EllComponent}, we have
\begin{align*}
\partial_\mathfrak{j}\bu_n^{[2i]}(\kp)=\partial_\mathfrak{j}[\bL^i\bu_n(\kp)]^{(0)}=[\bL^i\partial_\mathfrak{j}\bu_n(\kp)]^{(0)}
\end{align*}
and
\begin{align*}
\partial_\mathfrak{j}\bu_n^{[2i+1]}(\kp)=\partial_\mathfrak{j}[\bLd\bL^i\bu_n(\kp)]^{(0)}=[\bLd\bL^i\partial_\mathfrak{j}\bu_n(\kp)]^{(0)},
\end{align*}
which can be further expressed by algebraic expressions of $\bU_n^{[i,j]}$ and $\bu_n^{[i]}$
by replacing $\partial_\mathfrak{j}\bu_n(\kp)$ with the help of \eqref{uDyna} and \eqref{uDynb}.
For instance, we have for $i=0$ the simplest relations as follows:
\begin{align*}
&\partial_1\bu_n^{[0]}(\kp)=\bu_n^{[1]}(\kp)+\bU_n^{[0,2]}\bu_n^{[0]}(\kp), \\
&\partial_2\bu_n^{[0]}(\kp)=\left(1+\bU_n^{[0,1]}\right)\bu_n^{[2]}(\kp)-\bU_n^{[0,2]}\bu_n^{[1]}(\kp)+\bU_n^{[0,3]}\bu_n^{[0]}(\kp)-g\bu_n^{[-2]}(\kp), \\
&\bu_{n+1}^{[0]}(\kp)=\frac{1}{\sqrt{g}}\left(1+\bU_{n+1}^{[0,1]}\right)\bu_n^{[2]}(\kp), \\
&\bu_{n-1}^{[0]}(\kp)=\frac{1}{\sqrt{g}}\bU_{n-1}^{[0,2]}\bu_n^{[1]}(\kp)-\frac{1}{\sqrt{g}}\bU_{n-1}^{[0,3]}\bu_n^{[0]}(\kp)+\sqrt{g}\bu_n^{[-2]}(\kp).
\end{align*}
The formulae for the higher-order derivatives of $\bu_n^{[i]}$, such as $\partial_{\mathfrak{j}}^2\bu_n^{(i)}$, $\partial_{\mathfrak{j}}^3\bu_n^{(i)}$, etc.
for $\mathfrak{j}\in\mathbb{Z}^+$ and $i,j\in\mathbb{Z}$ are obtained by iteration.
Following such an algorithm, we are able to verify the linear equations in \eqref{EllKP:L}.
For example, by replacing all the derivatives by algebraic expressions of $\bU_n^{[i,j]}$ and $\bu_n^{[i]}(\kp)$, direct calculation shows that
\begin{align*}
&\partial_2\bu_n^{[0]}(\kp)+\partial_1^2\bu_n^{[0]}(\kp)-2\sqrt{g}v_nw_n\bu_{n+1}^{[0]}(\kp) \\
&+2\left(\bU_n^{[3,0]}+\bU_n^{[2,1]}-\bU_n^{[2,0]}\bU_n^{[2,0]}+\bU_n^{[2,2]}\bU_n^{[0,0]}-\frac{3}{2}e\right)\bu_n^{[0]}(\kp)
\end{align*}
is identically zero in virtue of the antisymmetry condition $\bU_n^{[i,j]}=-\bU_n^{[j,i]}$
as well as the algebraic relations (which helps us to eliminate those shifted variables $\bU_{n+1}^{[i,j]}$ in terms of $n$)
that follow from taking $[i,j]$-labels of \eqref{UDync}.
Such an identity is essentially the linear equation \eqref{EllKP:La} once $\bu_n^{[0]}(\kp)$ and $\bU_n^{[i,j]}$ are expressed by $\phi_n$, $u_n$, $v_n$ and $w_n$.
Linear equations \eqref{EllKP:Lb} and \eqref{EllKP:Lc} are proven similarly.
\end{proof}

\begin{remark}
The third equation in \eqref{EllKP:L} allows us to rewrite the other two equations as a two-component linear system composed of
\begin{align}\label{EllKP:Lax}
\partial_2
\begin{pmatrix}
\phi_n \\
\phi_{n+1}
\end{pmatrix}
=
\bP
\begin{pmatrix}
\phi_n \\
\phi_{n+1}
\end{pmatrix}
\quad \hbox{and} \quad
\partial_3
\begin{pmatrix}
\phi_n \\
\phi_{n+1}
\end{pmatrix}
=
\bQ
\begin{pmatrix}
\phi_n \\
\phi_{n+1}
\end{pmatrix},
\end{align}
where $\bP$ and $\bQ$ are $2\times2$ matrix operators given by
\bse\label{EllKP:LaxMat}
\begin{align}
\bP=
\begin{pmatrix}
-\partial_1^2-2\partial_1u_n+3e & 2\sqrt{g}v_n w_n \\
-2\sqrt{g} & \partial_1^2+2(\partial_1\ln v_n)\partial_1+2\partial_1u_n-3e+\partial_1^2\ln v_n+(\partial_1\ln v_n)^2-\partial_2\ln v_n
\end{pmatrix}
\end{align}
and
\begin{align}
\bQ=
\begin{pmatrix}
\partial_1^3+3(\partial_1u_n)\partial_1+\displaystyle\frac{3}{2}\left(\partial_1^2u_n-\partial_2u_n\right) & -3\sqrt{g}v_n(\partial_1 w_n) \\
-3\sqrt{g}\partial_1\ln v_n & *
\end{pmatrix},
\end{align}
\ese
respectively, in which
\begin{align*}
*={}&\partial_1^3+3(\partial_1\ln v_n)\partial_1^2+3\left(\partial_1u_n+\partial_1^2\ln v_n+(\partial_1\ln v_n)^2\right)\partial_1 \\
&+\frac{3}{2}\left(\partial_1^2u_n-\partial_2u_n\right)+3(\partial_1\ln v_n)(2\partial_1u_n-3e) \\
&+\frac{3}{2}(\partial_1+\partial_1\ln v_n)\left(\partial_1^2\ln v_n+(\partial_1\ln v_n)^2-\partial_2\ln v_n\right).
\end{align*}
The linear equations in \eqref{EllKP:Lax} form the Lax pair for the elliptic coupled \ac{KP} system, namely the zero curvature equation
\begin{align*}
\partial_3\bP-\partial_2\bQ+[\bP,\bQ]=0
\end{align*}
with $[\bP,\bQ]\doteq \bP\bQ-\bQ\bP$ yields the nonlinear system \eqref{EllKP:NL}.
\end{remark}

\section{Dimensional reductions}\label{S:Reduc}

In this section, we further discuss the relevant $(1+1)$-dimensional elliptic integrable models
by performing dimensional reductions on the elliptic coupled \ac{KP} system \eqref{EllKP:NL}.
This is realised by imposing a restriction of the form $\mathcal{F}(\kp,\kp')=0$ on the spectral variables in the linear integral equation \eqref{Integral},
such that the effective plane wave factor $\rho_n(\kp)\rho_n(\kp')$ turns out to be independent of a certain argument $x_j$.
Recall that the dynamics of the elliptic coupled \ac{KP} system completely relies on the effective plane wave factor
\begin{align}\label{EffPWF}
\rho_n(\kp)\rho_n(\kp')=\exp\left\{\sum_{j=0}^\infty\left(k^{2j+1}+k'^{2j+1}\right)x_{2j+1}
+\sum_{j=1}^\infty\left[K^j-\left(\frac{g}{K}\right)^j+K'^j-\left(\frac{g}{K'}\right)^j\right]x_{2j}\right\}
\left(\frac{KK'}{g}\right)^n.
\end{align}
Thus, we can impose
\begin{align*}
k^{2j_0+1}+k'^{2j_0+1}=0 \quad \Rightarrow \quad k+k'=0 \quad \hbox{or} \quad \sum_{i=0}^{2j_0}k^{2j_0-i}(-k')^i=0
\end{align*}
and
\begin{align*}
K^{j_0}-\left(\frac{g}{K}\right)^{j_0}+K'^{j_0}-\left(\frac{g}{K'}\right)^{j_0}=0 \quad \Rightarrow \quad
\left(\frac{KK'}{g}\right)^{j_0}=1 \quad \hbox{or} \quad K^{j_0}+K'^{j_0}=0,
\end{align*}
in order to realise $x_{2j_0+1}$- and $x_{2j_0}$-independence, respectively.
But in practice, we are only allowed to set
\begin{align}\label{Reduc:Odd}
\mathcal{F}_{2j_0+1}(\kp,\kp')\doteq\sum_{i=0}^{2j_0}k^{2j_0-i}(-k')^i=0
\end{align}
and
\begin{align}\label{Reduc:Even}
\mathcal{F}_{2j_0}(\kp,\kp')\doteq K^{j_0}+K'^{j_0}=0
\end{align}
to respectively perform $x_{2j_0+1}$- and $x_{2j_0}$-reductions.
This is because either $k+k'=0$ or $(KK'/g)^{j_0}=1$ will lead to trivialities;
to be more precise, useful independent variables are also eliminated in these two cases.
Below we shall give the simplest examples including the elliptic coupled \ac{KdV} and \ac{BSQ} systems
arising from the $x_2$- and $x_3$-reductions of \eqref{EllKP:NL}.

To construct the elliptic coupled \ac{KdV} system, we set
\begin{align*}
K+K'=0,
\end{align*}
namely the $j_0=1$ case of \eqref{Reduc:Even} leading to the $x_2$-independence.
This in turn implies
\begin{align*}
k^2+k'^2=6e
\end{align*}
due to the elliptic curve relation \eqref{Curve}.
The constraint on the spectral points $(k,K)$ and $(k',K')$ results in $\partial_{4j+2}\rho_n(\kp)\rho_n(\kp')=0$.
From the definition of $\bC_n$, i.e. \eqref{C}, we can easily prove $\partial_{4j+2}\bC_n=0$ and subsequently $\partial_{4j+2}\bU_n=0$ by \eqref{U}.
Recall that the $\tau$-function and the potentials are defined as \eqref{tau} and \eqref{Component}. We obtain reduction conditions
\begin{align*}
\partial_{4j+2}\tau_n=\partial_{4j+2}u_n=\partial_{4j+2}v_n=\partial_{4j+2}w_n=0,
\end{align*}
for $j=0,1,2,\cdots$.
Performing such reductions on the elliptic coupled \ac{KP} system \eqref{EllKP:NL}, we obtain a coupled system
\bse\label{EllKdV:NL}
\begin{align}
&\partial_3u_n=\frac{1}{4}\partial_1^3u_n+\frac{3}{2}(\partial_1u_n)^2+3g(1-v_nw_n), \\
&\partial_3v_n=-\frac{1}{2}\partial_1^3v_n-3(\partial_1u_n-3e)\partial_1v_n, \\
&\partial_3w_n=-\frac{1}{2}\partial_1^3w_n-3(\partial_1u_n-3e)\partial_1w_n,
\end{align}
\ese
which we shall refer to as the elliptic coupled \ac{KdV} system.
The bilinear form of equation \eqref{EllKdV:NL} is obtained from \eqref{EllKP:BL},
which is a multi-component system given by
\bse\label{EllKdV:BL}
\begin{align}
&\left(\rD_1^4-4\rD_1\rD_3\right)\tau_n\cdot\tau_n=24g\left(\tau_{n+1}\tau_{n-1}-\tau_n^2\right), \\
&\left(\rD_1^3+2\rD_3-18e\rD_1\right)\tau_{n+1}\cdot\tau_n=0, \\
&\left(\rD_1^3+2\rD_3-18e\rD_1\right)\tau_{n-1}\cdot\tau_n=0.
\end{align}
\ese
Notice that $\partial_{4j+2}\bU_n=0$ and $\bu(\kp)$ satisfies \eqref{u}. We can further derive
\begin{align*}
\partial_{4j+2}\phi_n=\left(K^{2j+1}-(g/K)^{2j+1}\right)\phi_n.
\end{align*}
This provides us with a reduction on the Lax pair \eqref{EllKP:Lax}.
Therefore, the Lax pair of the elliptic coupled \ac{KdV} system \eqref{EllKdV:NL} is composed of
\begin{align}\label{EllKdV:Lax}
\bP
\begin{pmatrix}
\phi_n \\
\phi_{n+1}
\end{pmatrix}
=
(K-g/K)
\begin{pmatrix}
\phi_n \\
\phi_{n+1}
\end{pmatrix}
\quad \hbox{and} \quad
\partial_3
\begin{pmatrix}
\phi_n \\
\phi_{n+1}
\end{pmatrix}
=
\bQ
\begin{pmatrix}
\phi_n \\
\phi_{n+1}
\end{pmatrix}.
\end{align}
Here the Lax matrices are given by
\bse\label{EllKdV:LaxMat}
\begin{align}
\bP=
\begin{pmatrix}
-\partial_1^2-2\partial_1u_n+3e & 2\sqrt{g}v_n w_n \\
-2\sqrt{g} & \partial_1^2+2(\partial_1\ln v_n)\partial_1+2\partial_1u_n-3e+\partial_1^2\ln v_n+(\partial_1\ln v_n)^2
\end{pmatrix}
\end{align}
and
\begin{align}
\bQ=
\begin{pmatrix}
\partial_1^3+3(\partial_1u_n)\partial_1+\displaystyle\frac{3}{2}\left(\partial_1^2u_n-\partial_2u_n\right) & -3\sqrt{g}v_n(\partial_1 w_n) \\
-3\sqrt{g}\partial_1\ln v_n & *
\end{pmatrix},
\end{align}
\ese
respectively, in which
\begin{align*}
*={}&\partial_1^3+3(\partial_1\ln v_n)\partial_1^2+3\left(\partial_1u_n+\partial_1^2\ln v_n+(\partial_1\ln v_n)^2\right)\partial_1 \\
&+\frac{3}{2}\partial_1^2u_n+3(\partial_1\ln v_n)(2\partial_1u_n-3e)+\frac{3}{2}(\partial_1+\partial_1\ln v_n)\left(\partial_1^2\ln v_n+(\partial_1\ln v_n)^2\right).
\end{align*}

For the elliptic coupled \ac{BSQ} system, we set $j_0=1$ in \eqref{Reduc:Odd}, namely
\begin{align*}
k^2-kk'+k'^2=0,
\end{align*}
in order to induce the $x_3$-independence.
This simultaneously leads to the fact that $K$ and $K'$ also obey a restriction in the form of
\begin{align*}
\left(K+3e+\frac{g}{K}\right)^2+\left(K+3e+\frac{g}{K}\right)\left(K'+3e+\frac{g}{K'}\right)+\left(K'+3e+\frac{g}{K'}\right)^2=0,
\end{align*}
which follows form the curve relation \eqref{Curve}.
In this case, the reduction conditions are given by
\begin{align*}
\partial_{3j}\tau_n=\partial_{3j}u_n=\partial_{3j}v_n=\partial_{3j}w_n=0 \quad \hbox{as well as} \quad
\partial_{3j}\phi_n=k^{3j}\phi_n
\end{align*}
for $j\in\mathbb{Z}^+$. Therefore, we obtain from \eqref{EllKP:NL} the elliptic coupled \ac{BSQ} system
\bse\label{EllBSQ:NL}
\begin{align}
&\partial_2^2u_n=-\frac{1}{3}\partial_1^4u_n-4(\partial_1u_n)\partial_1^2u_n+4g\partial_1(v_nw_n), \\
&\partial_1\partial_2v_n=\frac{1}{3}\partial_1^3v_n+2(\partial_1u_n-3e)\partial_1v_n-2(\partial_2u_n)v_n, \\
&\partial_1\partial_2w_n=-\frac{1}{3}\partial_1^3w_n-2(\partial_1u_n-3e)\partial_1w_n-2(\partial_2u_n)w_n.
\end{align}
\ese
The corresponding bilinear form is a consequence of the \eqref{EllKP:BL}, taking the form of the following:
\bse\label{EllBSQ:BL}
\begin{align}
&\left(\rD_1^4+3\rD_2^2\right)\tau_n\cdot\tau_n=24g\left(\tau_{n+1}\tau_{n-1}-\tau_n^2\right), \\
&\left(\rD_1^3-3\rD_1\rD_2-18e\rD_1\right)\tau_{n+1}\cdot\tau_n=0, \\
&\left(\rD_1^3+3\rD_1\rD_2-18e\rD_1\right)\tau_{n-1}\cdot\tau_n=0.
\end{align}
\ese
The Lax pair of the elliptic coupled \ac{BSQ} equation \eqref{EllBSQ:NL} is composed of
\begin{align}\label{EllBSQ:Lax}
\bQ
\begin{pmatrix}
\phi_n \\
\phi_{n+1}
\end{pmatrix}
=k^3
\begin{pmatrix}
\phi_n \\
\phi_{n+1}
\end{pmatrix}
\quad \hbox{and} \quad
\partial_2
\begin{pmatrix}
\phi_n \\
\phi_{n+1}
\end{pmatrix}
=
\bP
\begin{pmatrix}
\phi_n \\
\phi_{n+1}
\end{pmatrix},
\end{align}
where $\bP$ and $\bQ$ are the same as the ones given by \eqref{EllKP:LaxMat}.

We note that the third equations in \eqref{EllKdV:BL} and \eqref{EllBSQ:BL} can be omitted, since they are equivalent to their respective second equations.
But here we still reserve these equations in order to treat both \eqref{EllKdV:BL} and \eqref{EllBSQ:BL}
as $(1+1)$-dimensional three-component partial-differential systems of $\tau_n$, $\tau_{n+1}$ and $\tau_{n-1}$,
instead of $(2+1)$-dimensional differential-difference systems.

\section{Multi-soliton solutions to the elliptic coupled KP system}\label{S:Sol}

To further explore the integrability of the elliptic coupled \ac{KP} system,
we now specify a class of concrete solutions (i.e. elliptic multi-soliton solutions) to \eqref{EllKP:NL}.
They are constructed directly from the \ac{DL} framework by specifying a discrete measure
associated with distinct simple poles at values $\kp_\nu$ for $\nu=1,\dots,N$ of the uniformising spectral parameter,
associated with distinct points $(k_\nu,K_\nu)$ on the elliptic curve \eqref{Curve}.
This will lead the linear integral equation \eqref{Integral} to%
\begin{align}\label{DisIntegral}
\bu_n(\kp)+\sum_{\nu,\nu'=1}^N c_{\nu,\nu'}\rho_n(\kp)\Og(\kp,\kp_{\nu'})\rho_n(\kp_{\nu'})\bu_n(\kp_\nu)=\rho_n(\kp)\bc(\kp),
\end{align}
where the coefficients $c_{\nu,\nu'}$ are skew-symmetric in terms of the indices, i.e. $c_{\nu,\nu'}=-c_{\nu',\nu}$.
Let us introduce the finite matrices $\bM=(M_{\mu,\nu})_{N\times N}$ and $\bR=(R_{\mu,\nu})_{N\times N}$ whose respective entries are defined as
\begin{align}\label{M}
M_{\mu,\nu}\doteq\Og(\kp_\mu,\kp_\nu)=\frac{K_\mu-K_\nu}{k_\mu+k_\nu}=\frac{k_\mu-k_\nu}{1-\displaystyle\frac{g}{K_\mu K_\nu}}
\end{align}
and
\begin{align}\label{R}
R_{\mu,\nu}\doteq{}&c_{\nu,\mu}\rho_n(\kp_\nu)\rho_n(\kp_\mu) \nonumber \\
={}&c_{\nu,\mu}\exp\left\{\sum_{j=0}^\infty\left(k_\mu^{2j+1}+k_\nu^{2j+1}\right)x_{2j+1}
+\sum_{j=1}^\infty\left[K_\mu^j-\left(\frac{g}{K_\mu}\right)^j+K_\nu^j-\left(\frac{g}{K_\nu}\right)^j\right]x_{2j}\right\}
\left(\frac{K_\mu K_\nu}{g}\right)^n.
\end{align}
Equations \eqref{M} and \eqref{R} show that both $\bM$ and $\bR$ are skew-symmetric.
By setting $\kp=\kp_\mu$ for $\mu=1,2,\cdots,N$, we can rewrite \eqref{DisIntegral} as
\begin{align*}
\bphi_{\mu}+\sum_{\nu,\nu'=1}^N M_{\mu,\nu'}R_{\nu',\nu}\bphi_\nu=\bc(\kp_\mu),
\end{align*}
for the infinite vector $\bphi_\mu\doteq\bu_n(\kp_\mu)/\rho_n(\kp_\mu)$,
which can, in turn, be reformulated as a linear problem for an $N$-component block vector $(\bphi_1,\bphi_2,\cdots,\bphi_N)$ given by
\begin{align*}
(\bI+\bM\bR)
\begin{pmatrix}
\bphi_1 \\
\bphi_2 \\
\vdots \\
\bphi_N
\end{pmatrix}
=
\begin{pmatrix}
\bc(\kp_1) \\
\bc(\kp_2) \\
\vdots \\
\bc(\kp_N)
\end{pmatrix},
\end{align*}
namely
\begin{align}\label{MatIntegral}
(\bphi_1,\bphi_2,\cdots,\bphi_N)=(\bc(\kp_1),\bc(\kp_2),\cdots,\bc(\kp_N))(\bI+\bR\bM)^{-1},
\end{align}
where $\bI=\bI_{N\times N}$ denotes the $N\times N$ identity matrix.
Meanwhile, the discrete measure also reduces \eqref{Potential} to
\begin{align}\label{MatPotential}
\bU_n&=\sum_{\mu,\mu'=1}^N c_{\mu,\mu'}\bu_n(\kp_\mu)\tbc(\kp_{\mu'})\rho_n(\kp_{\mu'})
=\sum_{\mu,\mu'=1}^N \bphi_\mu R_{\mu',\mu}\tbc(\kp_{\mu'}) \nonumber \\
&=-(\bphi_1,\bphi_2,\cdots,\bphi_N)\bR
\begin{pmatrix}
\tbc(\kp_1) \\
\tbc(\kp_2) \\
\vdots \\
\tbc(\kp_N)
\end{pmatrix}.
\end{align}
Equation \eqref{MatPotential} together with \eqref{MatIntegral} allows us to write down the formula of $\bU_n$ for the elliptic multi-soliton solutions as follows:
\begin{align}\label{Sol:U}
\bU_n=-(\bc(\kp_1),\bc(\kp_2),\cdots,\bc(\kp_N))(\bI+\bR\bM)^{-1}\bR
\begin{pmatrix}
\tbc(\kp_1) \\
\tbc(\kp_2) \\
\vdots \\
\tbc(\kp_N)
\end{pmatrix}.
\end{align}
A remark here is that the multiplication in \eqref{Sol:U} is essentially based on $N\times N$ finite matrices $\bI$, $\bR$ and $\bM$
and $N$-component block vectors with their respective components being infinite vectors $\bc(\kp_i)$ and $\tbc(\kp_i)$.
Thus, the final result is an infinite matrix, which coincides with the form of \eqref{Potential}, as we expect.
Notice that the variables $u_n$, $v_n$ and $w_n$ of the elliptic coupled \ac{KP} system purely rely on the entries of $\bU_n$, see \eqref{Component}.
We can therefore derive the general formulae for the Cauchy matrix solutions to \eqref{EllKP:NL}.
\begin{theorem}
The elliptic coupled \ac{KP} system \eqref{EllKP:NL} possesses Cauchy matrix solutions
\bse\label{Sol:NL}
\begin{align}
&u_n=\tbk_2(\bI+\bR\bM)^{-1}\bR\,\bk_0, \\
&v_n=1-g^{-1}\,\tbk_3(\bI+\bR\bM)^{-1}\bR\,\bk_2, \\
&w_n=1+\tbk_1(\bI+\bR\bM)^{-1}\bk_0,
\end{align}
\ese
where $\bM$ and $\bR$ are $N\times N$ matrices introduced in \eqref{M} and \eqref{R},
and $\bk_i$ for $i=0,1,2,3$ are $N$-component column vectors defined as
\begin{align*}
&\bk_0={}^{t\!}(1,1,\cdots,1), \quad \bk_1={}^{t\!}(k_1,k_2,\cdots,k_N), \\
&\bk_2={}^{t\!}(K_1,K_2,\cdots,K_N) \quad \hbox{and} \quad \bk_3={}^{t\!}(k_1K_1,k_2K_2,\cdots,k_NK_N).
\end{align*}
\end{theorem}

For these solutions, the square of the $\tau$-function \eqref{tau} is of the form
\begin{align}\label{tauDet}
\tau_n^2=\det\left(\bI+\bM\bR\right)=\det\left(\bI+\bR\bM\right)=
\det
\left(
\begin{array}{c:c}
\bM & \bI \\
\hdashline
-\bI & \bR
\end{array}
\right).
\end{align}
Since both $\bM$ and $\bR$ are skew-symmetric matrices the latter $N\times N$ determinant is a square of a $(2N-1)\times(2N-1)$ Pfaffian,
and hence the $\tau$-function itself can be written as a Pfaffian\footnote{
We refer the reader to \cite{Mui04} for such a notation (i.e. the triangular array) of Pfaffian.
Here the postfactor is chosen to comply with a normalisation of the $\tau$-function such that it has the form $\tau_n=1+\textrm{perturbation}$.
}
that
\begin{align}\label{Sol:BL}
\tau_n&=\pf(\bM|\bR) \nonumber \\
&\doteq
\left.\begin{array}{ccccccccc}
|\,M_{1,2} & M_{1,3} & \cdots & M_{1,N} & 1 & 0 & \cdots & 0 & 0 \\
& M_{2,3} & \cdots & M_{2,N} & 0 & 1 & \cdots & 0 & 0 \\
& & \ddots & \vdots & \vdots & \vdots & \ddots & \vdots & \vdots \\
& & & M_{N-1,N} & 0 & 0 & \cdots & 1 & 0 \\
& & & & 0 & 0 & \cdots & 0 & 1 \\
& & & & & R_{1,2} & \cdots & R_{1,N-1} & R_{1,N} \\
& & & & & & \ddots & \vdots & \vdots \\
& & & & & & & R_{N-2,N-1} & R_{N-2,N} \\
& & & & & & & & R_{N-1,N}
\end{array}
\right|(-1)^{N(N-1)/2}.
\end{align}

In order to give an explicit expression for this Pfaffian (a definition is given in appendix \ref{S:Pfaff}),
we need an expansion formula for the Pfaffian,
similar to the expansion for a determinant of the form $\det(\boldsymbol{1}+\bM\bR)$ in terms of the matrix invariants of the matrix $\bM\bR$,
but now in terms of Pfaffians of $\bM$ and $\bR$ separately.
Such an expansion formula is given by \eqref{PfaffExp} in appendix \ref{S:Pfaff}.
Furthermore, we can make use of the lemma below which amounts to a Pfaffian analogue of the Frobenius (i.e. elliptic Cauchy, cf. \cite{Fro82}) determinant formula.
\begin{lemma}
The Pfaffian of an elliptic Cauchy matrix $\bM$ with entries $M_{i,j}=\frac{K_i-K_j}{k_i+k_j}$ for $i,j=1,2,\dots,m$,
where $(k_i,K_i)$ are distinct points on the elliptic curve \eqref{Curve}, is given by
\begin{align}\label{EllPhase}
\pf(\bM)=\frac{g^{m(m-2)/8}}{\left(\displaystyle\prod_{i=1}^m\,K_{i}\right)^{(m-2)/2}}\prod_{1\leq i<j\leq m}\frac{K_{i}-K_{j}}{k_{i}+k_{j}},
\end{align}
for $m$ is even, while the Pfaffian vanishes (by definition) for $m$ is odd.
\end{lemma}

\begin{remark}
Essentially formula \eqref{EllPhase} in the lemma has appeared in \cite{DJKM83} without a proof.
A proof based on the Frobenius determinant formula for the elliptic Cauchy matrix is given in appendix \ref{S:Cauchy}.
\end{remark}

\noindent
Together with the expansion formula \eqref{PfaffExp} this allows us to write an explicit `Hirota-type' formula\footnote{
Note that taking the square root of the corresponding expansion of the determinant expansion does not necessarily lead to a finite expansion formula,
while the Pfaffian analogue \eqref{PfaffExp} does provide one.
}
for the elliptic $N$-soliton solution of the elliptic coupled \ac{KP} system.
Combining the results of appendices \ref{S:Pfaff} and \ref{S:Cauchy}, we thus obtain the following explicit expression
for the $\tau$-function of the $N$-soliton solution to the elliptic coupled \ac{KP} system.
In fact, the expansion formula \eqref{PfaffExp} gives a finite sum of terms,
each of which is a product of individual sub-Pfaffians of the matrices $\bM$ and $\bR$.
The result is given in the following theorem.
\begin{theorem}
The $\tau$-function for the $N$-soliton solution to the bilinear elliptic coupled KP system \eqref{EllKP:BL} is given by
\begin{align}\label{tauExp}
\tau_n=1+\sum_{m\in J}\,\sum_{1\leq i_1<i_2<\cdots<i_m\leq N}
\frac{g^{m(m-2)/8}}{\left(\displaystyle\prod_{\nu=1}^m\,K_{i_\nu}\right)^{(m-2)/2}}\left(\prod_{1\leq\nu<\nu'\leq m}\frac{K_{i_\nu}-K_{i_{\nu'}}}{k_{i_{\nu}}+k_{i_{\nu'}}}\right)
\pf(\mathfrak{C}_{i_1,\cdots,i_m})\left(\prod_{\nu=1}^m\rho_n(\kp_{i_\nu})\right)
\end{align}
with $J=\{j=2i|i=1,2,\cdots,[N/2]\}$,
where $\mathfrak{C}_{i_1,\cdots,i_m}$ denotes the sub-matrix of the coefficient matrix $\mathfrak{C}=(c_{\mu,\nu})_{N\times N}$
obtained from selecting from it the rows and columns labelled by $i_1,\cdots,i_m$,
and the exponential factors $\rho_n$ are given in \eqref{PWF}.
The solutions $u_n$, $v_n$ and $w_n$ of the elliptic coupled KP system \eqref{EllKP:NL} are subsequently inferred from the expression \eqref{tauExp} into \eqref{uvwtau}.
\end{theorem}
\begin{remark}
Above we present the formula for the elliptic $N$-soliton solution parametrised by the curve \eqref{Curve}.
The degenerate solutions are obtained when any two of the branch points coincide with each other.
The situation when $e_1=e_2$ or $e_1=e_3$ corresponds to $g=0$ due to \eqref{Moduli}.
In these two cases, equation \eqref{EllKP:NL} reduces to the scalar \ac{KP} equation,
but simultaneously the elliptic $N$-soliton solution turns out to be a trivial one when $g=0$.
The third possibility $e_2=e_3$ leads to $9e^2=4g$, in which case we obtain a degenerate curve
\begin{align*}
k^2=K+3e+\frac{9}{4}\frac{e^2}{K}=\left(\sqrt{K}+\frac{3}{2}\frac{e}{\sqrt{K}}\right)^2
\end{align*}
as well as the corresponding degenerate solutions.
\end{remark}
\noindent
Below we give examples for $N=2,3,4,5$ ($N=1$ results in the seed solution $\tau_n=1$).
\begin{description}
\item[$N=2$]
\begin{align*}
\tau_n=1+\frac{K_1-K_2}{k_1+k_2}c_{1,2}\rho_n(\kp_1)\rho_n(\kp_2);
\end{align*}
\item[$N=3$]
\begin{align*}
\tau_n=1+\sum_{1\leq i_1<i_2\leq 3}\frac{K_{i_1}-K_{i_2}}{k_{i_1}+k_{i_2}}c_{i_1,i_2}\rho_n(\kp_{i_1})\rho_n(\kp_{i_2});
\end{align*}
\item[$N=4$]
\begin{align*}
\tau_n={}&1+\sum_{1\leq i_1<i_2\leq 4}\frac{K_{i_1}-K_{i_2}}{k_{i_1}+k_{i_2}}c_{i_1,i_2}\rho_n(\kp_{i_1})\rho_n(\kp_{i_2}) \\
&+\frac{g}{\displaystyle\prod_{i=1}^4K_i}\left(\prod_{1\leq i_1<i_2\leq 4}\frac{K_{i_1}-K_{i_2}}{k_{i_1}+k_{i_2}}\right)
(c_{1,2}c_{3,4}-c_{1,3}c_{2,4}+c_{1,4}c_{2,3})\left(\prod_{i=1}^4\rho_n(\kp_i)\right);
\end{align*}
\item[$N=5$]
\begin{align*}
\tau_n={}&1+\sum_{1\leq i_1<i_2\leq 5}\frac{K_{i_1}-K_{i_2}}{k_{i_1}+k_{i_2}}c_{i_1,i_2}\rho_n(\kp_{i_1})\rho_n(\kp_{i_2}) \\
&+\sum_{1\leq i_1<i_2<i_3<i_4\leq 5}\frac{g}{\displaystyle\prod_{\nu=1}^4K_{i_\nu}}\left(\prod_{1\leq\nu<\nu'\leq 4}\frac{K_{i_\nu}-K_{i_{\nu'}}}{k_{i_{\nu}}+k_{i_{\nu'}}}\right) \\
&\qquad\qquad\qquad\qquad\times(c_{i_1,i_2}c_{i_3,i_4}-c_{i_1,i_3}c_{i_2,i_4}+c_{i_1,i_4}c_{i_2,i_3})\left(\prod_{\nu=1}^4\rho_n(\kp_{i_\nu})\right).
\end{align*}
\end{description}

We observe from the $\tau$-function \eqref{tauExp} as well as the examples that
the solutions for even $N$ resemble those of the BKP hierarchy apart from extra terms as well as the elliptic kernel and plane wave factors.
For instance, setting $c_{1,3}=c_{2,4}=c_{1,4}=c_{2,3}=0$ and introducing $d_{i,j}$ (which are effectively arbitrary constants) determined by
\begin{align*}
\exp(d_{i,j})\doteq c_{i,j}\frac{K_i-K_j}{k_i+k_j}
\end{align*}
in the formula for $N=4$ yields a particular solution to the bilinear system \eqref{EllKP:BL} as follows:
\begin{align*}
\tau_n={}&1+\exp(\xi_1+\xi_2+d_{1,2})+\exp(\xi_3+\xi_4+d_{3,4}) \\
&+\frac{g}{K_1K_2K_3K_4}\frac{(K_1-K_3)(K_1-K_4)(K_2-K_3)(K_2-K_4)}{(k_1+k_3)(k_1+k_4)(k_2+k_3)(k_2+k_4)}\exp(\xi_1+\xi_2+\xi_3+\xi_4+d_{1,2}+d_{3,4}),
\end{align*}
in which the plane wave factors are given by
\begin{align*}
\exp(\xi_i)\doteq\exp\left\{\sum_{j=0}^\infty k_i^{2j+1}x_{2j+1}+\sum_{j=1}^\infty\left[K_i^j-\left(\frac{g}{K_i}\right)^j\right]x_{2j}\right\}\left(\frac{K_i}{\sqrt{g}}\right)^n
\end{align*}
for $i=1,2,3,4$. While the `two-soliton' solution to the bilinear BKP hierarchy, including its first member
\begin{align*}
(\rD_1^6-5\rD_1^3\rD_3-5\rD_3^2+9\rD_1\rD_5)\tau\cdot\tau=0,
\end{align*}
takes the form of (cf. \cite{DJKM2} and also\cite{Hir89a})
\begin{align*}
\tau={}&1+\exp(\xi_1+\xi_2+d_{1,2})+\exp(\xi_3+\xi_4+d_{3,4}) \\
&+\frac{(k_1-k_3)(k_1-k_4)(k_2-k_3)(k_2-k_4)}{(k_1+k_3)(k_1+k_4)(k_2+k_3)(k_2+k_4)}\exp(\xi_1+\xi_2+\xi_3+\xi_4+d_{1,2}+d_{3,4})
\end{align*}
with plane wave factors
\begin{align*}
\exp(\xi_i)\doteq\exp\left\{\sum_{j=0}^\infty k_i^{2j+1}x_{2j+1}\right\}
\end{align*}
for $i=1,2,3,4$ and arbitrary constants $d_{1,2}$ and $d_{3,4}$.
The main difference here is that the linear dispersion is parametrised by the elliptic curve \eqref{Curve},
and also, an elliptic phase shift term (as a consequence of the significant formula \eqref{EllPhase}) is involved.
They together describe an elliptic-type (which is remarkable from our viewpoint) soliton interaction.
The formulae for odd $N$ are in a sense the respective parameter extensions of those for even numbers $N-1$.

\section{Multi-soliton solutions to the DKP equation with nonzero constant background}\label{S:DKP}

We now discuss the connection between the elliptic coupled \ac{KP} system \eqref{DJM} and the DKP equation \eqref{DKP},
and show how soliton solutions to the DKP equation \eqref{DKP} with nonzero constant background are constructed as a byproduct of the result in section \ref{S:Sol}.

By introducing new variables
\begin{align}\label{TransScal}
\cU\doteq2\partial_1u+c_1, \quad \cV\doteq c_2 v \quad \hbox{and} \quad \cW\doteq c_3 w,
\end{align}
where $c_1$, $c_2$ and $c_3$ are constants obeying $c_1=-6e$ and $c_2c_3=g$, we are able to reformulate \eqref{DJM} as
\bse\label{NPEllKP}
\begin{align}
&\partial_3\cU=\frac{1}{4}\partial_1^3\cU+\frac{3}{2}(\cU+6e)\partial_1\cU+\frac{3}{4}\partial^{-1}_1\partial_2^2\cU-6\partial_1(\cV\cW), \\
&\partial_3\cV=-\frac{1}{2}\partial_1^3\cV-\frac{3}{2}\cU\partial_1\cV+\frac{3}{2}\partial_1\partial_2\cV+\frac{3}{2}(\partial_1^{-1}\partial_2\cU)\cV, \\
&\partial_3\cW=-\frac{1}{2}\partial_1^3\cW-\frac{3}{2}\cU\partial_1\cW-\frac{3}{2}\partial_1\partial_2\cW-\frac{3}{2}(\partial_1^{-1}\partial_2\cU)\cW.
\end{align}
\ese
Then the following transformations between partial-differential operators:
\begin{align*}
\partial_1=\partial_x, \quad \partial_2=\partial_y, \quad \partial_3=\partial_t+9e\partial_x,
\end{align*}
which effectively follow from a Galilean transformation composed of
\begin{align}\label{Galilean}
x\doteq x_1+9e x_3, \quad y\doteq x_2 \quad \hbox{and} \quad t\doteq x_3,
\end{align}
leads \eqref{NPEllKP} to the DKP equation \eqref{DKP}.
Notice that the solutions we have obtained for the elliptic coupled \ac{KP} system in section \ref{S:Sol} are the ones with background $u=0$, $v=1$ and $w=1$.
This implies that here we are able to construct soliton solutions to the DKP equation with nonzero constant background with the help of \eqref{TransScal} and \eqref{Galilean}.
We conclude this as the theorem below.
\begin{theorem}
The $N$-soliton solutions to the DKP equation \eqref{DKP} with nonzero constant background\footnote{
The constant $c_1$ is not necessarily nonzero, because $c_1=0$ does not lead to degeneration of the elliptic curve \eqref{DKP:Curve}.
}
\begin{align}\label{DKP:BG}
\cU=c_1, \quad \cV=c_2, \quad \cW=c_3 \quad \hbox{for} \quad c_2c_3\neq 0
\end{align}
are given by
\begin{align}\label{DKP:Sol}
\cU=c_1+2\partial_x^2\ln\tau_n, \quad \cV=c_2\frac{\tau_{n+1}}{\tau_n} \quad \hbox{and} \quad \cW=c_3\frac{\tau_{n-1}}{\tau_n},
\end{align}
where the $\tau$-function (cf. \eqref{tauExp}) takes the form of
\begin{align*}
\tau_n=1+\sum_{m\in J}\,\sum_{1\leq i_1<i_2<\cdots<i_m\leq N}
\frac{(c_2c_3)^{m(m-2)/8}}{\left(\displaystyle\prod_{\nu=1}^m\,K_{i_\nu}\right)^{(m-2)/2}}\left(\prod_{1\leq\nu<\nu'\leq m}\frac{K_{i_\nu}-K_{i_{\nu'}}}{k_{i_{\nu}}+k_{i_{\nu'}}}\right)
\pf(\mathfrak{C}_{i_1,\cdots,i_m})\left(\prod_{\nu=1}^m\rho_n(\kp_{i_\nu})\right)
\end{align*}
with $(k_{i_{\nu}},K_{i_\nu})$ being the spectral points on the elliptic curve
\begin{align}\label{DKP:Curve}
k^2=K-\frac{c_1}{2}+\frac{c_2c_3}{K}
\end{align}
and $\rho_n(\kp_{i_\nu})$ being the plane wave factors determined by
\begin{align}\label{DKP:PWF}
\rho_n(\kp)=\exp\left\{kx+\left(K-\frac{c_2c_3}{K}\right)y+\left(k^3+\frac{3c_1}{2}k\right)t\right\}\left(\frac{K}{\sqrt{c_2c_3}}\right)^n.
\end{align}
\end{theorem}

\begin{remark}
The special case when the background constants $c_1$, $c_2$ and $c_3$ satisfy $c_1^2=16c_2c_3$
will lead to the fact that the spectral points in the soliton solutions are parametrised by the degenerate curve
\begin{align*}
k^2=K-\frac{c_1}{2}+\frac{c_1^2}{16}\frac{1}{K}=\left(\sqrt{K}-\frac{c_1}{4}\frac{1}{\sqrt{K}}\right)^2.
\end{align*}
We comment that the multi-soliton solutions to the DKP equation discussed here
differ from those with zero background $\cU=\cV=\cW=0$ in the literature, cf. \cite{HO91,Gil02,KM06}.
In this paper, the constants $c_1$, $c_2$ and $c_3$ for the nonzero `seed solutions' will eventually
play a role of the moduli in the elliptic curve that parameterises the spectral points in the multi-solitons.
\end{remark}

\noindent
For example, by taking $N=2$ we obtain the simplest nontrivial $\tau$-function for soliton solutions to the DKP equation \eqref{DKP}
with nonzero constant background \eqref{DKP:BG} as follows:
\begin{align*}
\tau_n=1+\exp\Bigg\{(k_1+k_2)x&+\left(K_1-\frac{c_2c_3}{K_1}+K_2-\frac{c_2c_3}{K_2}\right)y \\
&+\left(k_1^3+k_2^3+\frac{3c_1}{2}(k_1+k_2)\right)t+d_{1,2}\Bigg\}\left(\frac{K_1K_2}{c_2c_3}\right)^n,
\end{align*}
where $(k_1,K_1)$ and $(k_2,K_2)$ are points on the elliptic curve \eqref{DKP:Curve} and $d_{1,2}$ is an arbitrary constant.
In this case, the corresponding components for the `one-soliton' solution to the DKP equation \eqref{DKP} are given by
\begin{align*}
\cU=c_1+2\partial_x^2\ln\tau_0, \quad \cV=c_2\frac{\tau_1}{\tau_0} \quad \hbox{and} \quad  \cW=c_3\frac{\tau_{-1}}{\tau_0},
\end{align*}
where we have fixed $n=0$ in \eqref{DKP:Sol} for simplicity.

We can similarly construct solutions with nonzero constant background given by $\cU=c_1$, $\cV=c_2$ and $\cW=c_3$ to the coupled \ac{KdV} system
(which was introduced in \cite{WGHZ99} as a generalisation of the Hirota--Satsuma equation \cite{HS81})
\bse\label{CKdV}
\begin{align}
&\partial_t\cU=\frac{1}{4}\partial_x^3\cU+\frac{3}{2}\cU\partial_x\cU-6\partial_x(\cV\cW), \\
&\partial_t\cV=-\frac{1}{2}\partial_x^3\cV-\frac{3}{2}\cU\partial_x\cV, \\
&\partial_t\cW=-\frac{1}{2}\partial_x^3\cW-\frac{3}{2}\cU\partial_x\cW
\end{align}
\ese
based on \eqref{EllKdV:NL}. Since the procedure is the same, we omit the detail here.

Note that Galilean transformation \eqref{Galilean} has affected the coefficient of $t$ in the linear dispersion (compare \eqref{DKP:PWF} with \eqref{PWF}).
Generally we are not able to directly construct multi-soliton solutions to the coupled \ac{BSQ} system (which follows from the $t$-independent reduction of \eqref{DKP})
\bse\label{CBSQ}
\begin{align}
&\partial_y^2\cU=-\frac{1}{3}\partial_1^4\cU-2\cU\partial_1^2\cU-2(\partial_1\cU)^2+8\partial_1^2(\cV\cW), \\
&\partial_x\partial_y\cV=\frac{1}{3}\partial_x^3\cV+\cU\partial_x\cV-(\partial_x^{-1}\partial_y\cU)\cV, \\
&\partial_x\partial_y\cW=-\frac{1}{3}\partial_x^3\cW-\cU\partial_x\cW-(\partial_x^{-1}\partial_y\cU)\cW
\end{align}
\ese
with nonzero constant background from those to \eqref{EllBSQ:NL},
unless we discuss a very special case for $c_1=0$ (which implies that the constant background is given by $\cU=0$, $\cV=c_2$ and $\cW=c_3$).
This is because the $t$-independent reduction from \eqref{DKP} to \eqref{CBSQ} requires a constraint in the form of
\begin{align*}
\mathcal{F}(\kp,\kp')=k^3+k'^3+\frac{3c_1}{2}(k+k')=0.
\end{align*}
which is incompatible with the constraint $k^2+kk'+k'^2=0$ (essentially leading to $k^3+k'^3=0$) in the construction of \eqref{EllBSQ:NL},
from the perspective of the \ac{DL}, cf. section \ref{S:Reduc}.

\section{Concluding remarks}\label{S:Concl}

The elliptic coupled \ac{KP} system \eqref{DJM} was studied within the \ac{DL} framework,
which provided us with a unified perspective to understand the integrability of the elliptic model.
As a consequence, we have constructed its Lax pair and elliptic soliton solutions.
The elliptic coupled \ac{KdV} and \ac{BSQ} systems were obtained
from dimensional reductions of the elliptic coupled \ac{KP} system, together with their respective Lax pairs.
An interesting observation is that from the elliptic coupled \ac{KP} system we are able to construct a new class of solutions
(i.e. solitons with nonzero constant background) to the DKP equation through a Galilean transformation.

Since the \ac{DL} approach is akin to the Riemann--Hilbert method appearing in the inverse scattering, based on similar singular integral equations,
we expect that the present results open the way to a comprehensive study of the initial value problems associated with these elliptic models.
Also, we expect that algebro-geometric solutions of higher-genus can be treated,
as well as other reductions (e.g. pole reductions) leading to possibly novel finite-dimensional integrable systems.
Furthermore, the structures emerging from the \ac{DL} approach in terms of the infinite matrix representation \eqref{UDyn} are evidence for the possibility
that these nonlinear systems are also integrable in the sense of possessing higher-order symmetries.
However, it remains an open problem to explicitly construct these higher-order symmetries.
In particular, we may want to establish a Sato-type scheme for the coupled elliptic \ac{KP} system through a matrix pseudo-differential operator algebra,
and construct recursion operators for the elliptic coupled \ac{KdV} and \ac{BSQ} systems.

Further aspects of the elliptic \ac{KP} family of systems remain to be investigated.
So far, we have only constructed the single-component \ac{KP}-type equation and its dimensional reductions in the present paper.
It would be interesting to find the multi-component (or matrix) \ac{KP}-type hierarchy in this family,
and simultaneously to see how they are related to the \ac{LL} equation.
In addition to the \ac{LL} equations, there also exist the Krichever--Novikov equation \cite{KN80} and the elliptic analogue of the Toda equation \cite{Kri00}.
Both are elliptic integrable systems that play roles of the master equations in their respective classes.
How these particular elliptic systems and their discrete analogues are related to this elliptic family remains a problem for future work.
The first step towards this goal might be searching for a discrete analogue of the elliptic coupled \ac{KP} system \eqref{DJM},
since the integrability of the discrete (non-elliptic) coupled \ac{KP} system has been studied in \cite{GNT01}
which we believe would bring us insights into the elliptic case.

\section*{Acknowledgments}
WF was supported by the National Natural Science Foundation of China (grant no. 11901198)
and the Science and Technology Commission of Shanghai Municipality (grant no. 18dz2271000).
FWN was supported by the Engineering and Physical Sciences Research Council (EP/W007290/1).
\begin{appendix}

\section{Bilinear derivative and logarithmic transformations}\label{S:BL}

In this appendix we briefly recapitulate the notion of Hirota's bilinear derivative and the relevant logarithmic transformations.
We also refer the reader to the monographs \cite{Hir04} and \cite{MJD00} for more details regarding bilinear derivatives.
\begin{definition}
Suppose that $F$ and $G$ are differentiable functions of the independent variables $x_j$ for $j\in\mathbb{Z}^+$.
The $m$th-order bilinear derivative of $F$ and $G$ with respect to the argument $x_j$ is defined as
\begin{align}
\rD_j^m F\cdot G\doteq \left(\frac{\partial}{\partial x_j}-\frac{\partial}{\partial x'_j}\right)^m F(\bx)G(\bx')\Bigg|_{\bx'=\bx}
\end{align}
for $\bx=(x_1,x_2,\cdots)$ and $\bx'=(x'_1,x'_2,\cdots)$.
\end{definition}

\begin{remark}
We can alternatively define the bilinear derivative by
\begin{align}
\re^{\varepsilon\rD_j}F\cdot G=F(\cdots,x_{j-1},x_j+\varepsilon,x_{j+1},\cdots)G(\cdots,x_{j-1},x_j-\varepsilon,x_{j+1},\cdots),
\end{align}
in which $\varepsilon$ is a parameter.
Then we obtain the explicit formulae for $\rD_j^m F\cdot G$ from the coefficients of $\varepsilon^m$ for $m=1,2,\cdots$ in the series expansion.
\end{remark}

Bilinear derivatives are closely related to derivatives of logarithmic functions through the so-called logarithmic and bi-logarithmic transformations.
These transformations are derived from some fundamental identities in terms of exponents of $\rD$- and $\partial$-operators, see e.g. \cite{Hir04}.
Below we only list transformations that are needed in this paper (which can even be proven by definition of bilinear derivative).
The logarithmic transformations include the following:
\bse\label{Log}
\begin{align}
&\frac{\rD_j^2 F\cdot F}{F^2}=2\partial_j^2\ln F, \\
&\frac{\rD_j^4 F\cdot F}{F^2}=2\partial_j^4\ln F+12\left(\partial_j^2\ln F\right)^2, \\
&\frac{\rD_j^6 F\cdot F}{F^2}=2\partial_j^6\ln F+60\left(\partial_j^2\ln F\right)\left(\partial_j^4\ln F\right)+120\left(\partial_j^2\ln F\right)^3, \\
&\frac{\rD_i\rD_j F\cdot F}{F^2}=2\partial_i\partial_j\ln F.
\end{align}
\ese
These transformations only involve a single function $F$.
Instead, bi-logarithmic transformations involve two functions $F$ and $G$.
The first few of such transformations are as follows:
\bse\label{BiLog}
\begin{align}
&\frac{\rD_j F\cdot G}{FG}=\partial_j\ln\frac{F}{G}, \\
&\frac{\rD_j^2 F\cdot G}{FG}=\left(\partial_j\ln\frac{F}{G}\right)^2+\partial_j^2\ln\frac{F}{G}+2\partial_j^2\ln G, \\
&\frac{\rD_j^3 F\cdot G}{FG}=\left(\partial_j\ln\frac{F}{G}\right)^3+\partial_j^3\ln\frac{F}{G}
+3\left(\partial_j\ln\frac{F}{G}\right)\left(\partial_j^2\ln\frac{F}{G}+2\partial_j^2\ln G\right), \\
&\frac{\rD_iD_j F\cdot G}{FG}=\left(\partial_i\ln\frac{F}{G}\right)\left(\partial_j\ln\frac{F}{G}\right)+\partial_i\partial_j\ln\frac{F}{G}+2\partial_i\partial_j\ln G.
\end{align}
\ese
Equations \eqref{Log} and \eqref{BiLog} together allow us to transfer the bilinear equations in \eqref{EllKP:BL} to nonlinear equations in \eqref{EllKP:NL}.
We can also use these formulae reversely, in order to reformulate the nonlinear system as its corresponding bilinear form.

\section{Pfaffian and an expansion formula}\label{S:Pfaff}

Here we remind the reader of a few facts about Pfaffians (see \cite{Mui04,Hir04}), and present an expansion formula (see \cite{Cai59,Oka19})
that is effective in expressing the soliton solutions of the elliptic coupled \ac{KP} system in a concise (Hirota-type) form.

Let $\bA$ be an $N\times N$ skew-symmetric matrix with entries $a_{i,j}$, hence $\bA$
is of the form
\begin{align*}
\bA=
\begin{pmatrix}
0 & a_{1,2} & a_{1,3} & \cdots & a_{1,N} \\
-a_{1,2} & 0 & a_{2,3} & \cdots & a_{2,N} \\
-a_{1,3} & -a_{2,3} & 0 & \cdots & \vdots \\
\vdots & \vdots & \vdots & \ddots & a_{N-1,N} \\
-a_{1,N} & -a_{2,N} & \cdots & -a_{N-1,N} & 0
\end{pmatrix}.
\end{align*}
The Pfaffian $\pf(\bA)$ associated with the matrix $\bA$ can be defined as follows.
\begin{definition}
The $N$th-order Pfaffian associated with $\bA$ is the triangular array (see e.g. \cite{Mui04} for such a notation)
\bse\label{PfaffDef}
\begin{align}\label{PfaffForm}
\pf(\bA)=
\left.
\begin{array}{cccc}
|\, a_{1,2} & a_{1,3} & \cdots & a_{1,N} \\
& a_{2,3} & \cdots & a_{2,N} \\
& & \ddots & \vdots \\
& & & a_{N-1,N}
\end{array}
\right|,
\end{align}
which is uniquely defined by the recursion relation
\begin{align} \label{PfaffRecurs}
\pf(\bA)=\sum_{i=2}^N(-1)^i a_{1,i}
\left.
\begin{array}{ccccccc}
|\,a_{2,3}& \cdots &a_{2,i-1} & a_{2,i+1} & a_{2,i+2} & \cdots & a_{1,N} \\
& \ddots & \vdots & \vdots & \vdots & \vdots & \vdots \\
& & a_{i-2,i-1} & a_{i-2,i+1} & a_{i-2,i+2} & \cdots & a_{i-2,N} \\
& & & a_{i-1,i+1} & a_{i-1,i+2} & \cdots & a_{i-1,N} \\
& & & & a_{i+1,i+2} & \cdots & a_{i+1,N} \\
& & & & & \ddots & \vdots \\
& & & & & & a_{N-1,N}
\end{array}
\right|,
\end{align}
\ese
together with the initial values defined as $|\,\cdot\,|\doteq 0$ and $|a_{1,2}|\doteq a_{1,2}$ for $N=1$ and $N=2$, respectively.
\end{definition}

\begin{remark}
The definition implies that a Pfaffian is only nonzero when $N$ is even.
For example, we have
\begin{align*}
\left.
\begin{array}{cc}
|\, a_{1,2} & a_{1,3} \\
& a_{2,3}
\end{array}
\right|=0 \quad \hbox{and} \quad
\left.
\begin{array}{ccc}
|\, a_{1,2} & a_{1,3} & a_{1,4} \\
& a_{2,3} & a_{2,4} \\
& & a_{3,4}
\end{array}
\right|
=a_{1,2}a_{3,4}-a_{1,3}a_{2,4}+a_{1,4}a_{2,3},
\end{align*}
for $N=3$ and $N=4$, respectively.
We note the following important relation between the Pfaffians and determinants of skew-symmetric matrices of the form given above:
\begin{align}\label{DetPfaff}
\det(\bA)=\left(\pf(\bA)\right)^2.
\end{align}
Hence, the Pfaffian of $\bA$ can be thought of as the square root of a determinant of a skew-symmetric matrix.
However, the latter relation does not define the Pfaffian uniquely, while the recursion relation does.
\end{remark}

What we need, in order to obtain our explicit form of the soliton solutions,
is a Pfaffian analogue of the expansion formula for a determinant of the type $\det(\bI+\bA\bB)$
in terms of the matrix invariants of $\bA\bB$ (which are the sums of its principal minors).
This is given in the lemma below.
To express the formula in a compact manner, let us introduce the notation
\begin{align*}
\bA_{i_1,i_2,\cdots,i_m}
\end{align*}
denoting the sub-matrix of $\bA$ by selecting from it the rows and columns labelled by $i_1,i_2,\cdots,i_m$ for $1\leq i_1<i_2<\cdots<i_m<N$.
For examples, we have
\begin{align*}
\bA_{i_1}=a_{i_1,i_1}, \quad
\bA_{i_1,i_2}=
\begin{pmatrix}
a_{i_1,i_1} & a_{i_1,i_2} \\
a_{i_2,i_1} & a_{i_2,i_2}
\end{pmatrix}
\quad \hbox{and} \quad
\bA_{1,2,\cdots,N}=\bA.
\end{align*}
In this notation, the expansion formula \eqref{PfaffRecurs} can be rewritten as
\begin{align*}
\pf(\bA)=\sum_{i=2}^N(-1)^i a_{1,i}\pf(\bA_{2,\cdots,i-1,i+1,\cdots,N}).
\end{align*}
\begin{lemma}
Let $\bA$ and $\bB$ be $N\times N$ skew-symmetric matrices of the form given above for $\bA$ and a similar form for $\bB$ (with entries $b_{i,j}$).
The following expansion formula holds for the special Pfaffian of the format
\begin{align}\label{PfaffExp}
\pf(\bA|\bB)\doteq{}&
\left.\begin{array}{ccccccccc}
|\,a_{1,2} & a_{1,3} & \cdots & a_{1,N} & 1 & 0 & \cdots & 0 & 0 \\
& a_{2,3} & \cdots & a_{2,N} & 0 & 1 & \cdots & 0 & 0 \\
& & \ddots & \vdots & \vdots & \vdots & \ddots & \vdots & \vdots \\
& & & a_{N-1,N} & 0 & 0 & \cdots & 1 & 0 \\
& & & & 0 & 0 & \cdots & 0 & 1 \\
& & & & & b_{1,2} & \cdots & b_{1,N-1} & b_{1,N} \\
& & & & & & \ddots & \vdots & \vdots \\
& & & & & & & b_{N-2,N-1} & b_{N-2,N} \\
& & & & & & & & b_{N-1,N}
\end{array}
\right| \nonumber \\
={}&(-1)^{N(N-1)/2}\left[1+\sum_{m\in J}(-1)^{m(m-1)/2}\sum_{1\leq i_1<i_2<\cdots<i_m\leq N}
\pf(\bA_{i_1,i_2,\cdots,i_m})\pf(\bB_{i_1,i_2,\cdots,i_m})\right],
\end{align}
for $J=\{j=2i|i=1,2,\cdots,[N/2]\}$.
\end{lemma}
\noindent
Formula \eqref{PfaffExp} is a special case of a Laplace-type expansion formula for Pfaffian,
see proposition 2.3 in \cite{Oka19} for the general formula and its proof.
The crucial upshot of the lemma is that the terms in the expansion \eqref{PfaffExp} are products of separate sub-Pfaffians of $\bA$ and $\bB$, respectively.
To give an idea how this expansion looks like, let us write them down explicitly for the values of $N=2,3,4,5$.
\begin{description}
\item[$N=2$]
\begin{align*}
\pf(\bA\,|\,\bB)=-1+a_{1,2}b_{1,2};
\end{align*}
\item[$N=3$]
\begin{align*}
\pf(\bA\,|\,\bB)=-1+a_{1,2}b_{1,2}+a_{1,3}b_{1,3}+a_{2,3}b_{2,3};
\end{align*}

\item[$N=4$]
\begin{align*}
\pf(\bA\,|\,\bB)=1-\sum_{1\leq i_1<i_2\leq 4} a_{i_1,i_2} b_{i_1,i_2}+
\left. \begin{array}{ccc}
|\, a_{1,2} & a_{1,3} &  a_{1,4} \\
& a_{2,3} & a_{2,4} \\
& & a_{3,4}
\end{array}
\right|\,
\left.
\begin{array}{ccc}
|\, b_{1,2} & b_{1,3} &  b_{1,4} \\
& b_{2,3} & b_{2,4} \\
& & b_{3,4}
\end{array}
\right|;
\end{align*}

\item[$N=5$]
\begin{align*}
\pf(\bA\,|\,\bB)={}1&-\sum_{1\leq i_1<i_2\leq 5}a_{i_1,i_2}b_{i_1,i_2} \\
&+\sum_{1\leq i_1<i_2<i_3<i_4\leq 5}
\left.
\begin{array}{ccc}
|\, a_{i_1,i_2} & a_{i_1,i_3} & a_{i_1,i_4} \\
& a_{i_2,i_3} & a_{i_2,i_4} \\
& & a_{i_3,i_4}
\end{array}
\right|\,
\left.
\begin{array}{ccc}
|\, b_{i_1,i_2} & b_{i_1,i_3} & b_{i_1,i_4} \\
& b_{i_2,i_3} & b_{i_2,i_4} \\
& & b_{i_3,i_4}
\end{array}
\right|.
\end{align*}
\end{description}
In the case of the soliton solution \eqref{Sol:BL} the fact that in this expansion we get sums of products of separate Pfaffians is crucial,
as we can compute the Pfaffians of the elliptic Cauchy matrix in explicit form, with the formulae given in the next appendix.

\section{Applying Frobenius formula for elliptic Cauchy matrix}\label{S:Cauchy}

We introduce the elliptic functions
\begin{align}\label{W}
W_1(x)\equiv\Phi_{\og_1}(x)\re^{-\eta_1x}, \quad W_2(x)\equiv\Phi_{\og_2}(x)\re^{-\eta_2x} \quad \hbox{and} \quad W_3(x)\equiv\Phi_{\og_3}(x)\re^{-\eta_3x},
\end{align}
where $\Phi_\kappa(x)$ is the Lam\'e function, given in terms of the Weierstrass
$\sigma$-function $\sg(x)=\sg\left(x|2\og_1,2\og_2\right)$ with half-periods $\og_1,\og_2$
($2\og_1$ and $2\og_2$ being the elementary periods of the period lattice):
\begin{align}\label{Phi}
\Phi_\alpha(x)\equiv\frac{\sg(x+\alpha)}{\sg(\alpha)\sg(x)},
\end{align}
where $\alpha$ is an arbitrary complex variable not coinciding with any zero of the
$\sg$-function, and $\eta_1=\zeta(\og_1)$, $\eta_2=\zeta(\omega_2)$, $\eta_3=\zeta(\omega_3)$,
where $\og_3=-\og_1-\og_2$ (see e.g. \cite{Akh90} for the standard notation of Weierstrass elliptic functions).

From the standard addition formulae for the Weierstrass functions $\sg(x)$, $\zeta(x)$ and $\wp(x)$
we obtain the relations for the functions $W_1$, $W_2$ and $W_3$ including the Yang--Baxter-type relation
\begin{equation}\label{WW}
W_1(x)W_2(z)+W_2(y)W_3(x)+W_3(z)W_1(y)=0, \quad x+y+z=0.
\end{equation}
Equation \eqref{WW} follows from the well-known 3-term addition formula for the $\sg$-function,
which can be written in the compact form of an elliptic partial fraction expansion in terms of $\Phi$, namely
\begin{align}\label{PhiAdd}
\Phi_\alpha(x)\Phi_\beta(y)=\Phi_{\alpha+\beta}(x)\Phi_\beta(y-x)+\Phi_\alpha(x-y)\Phi_{\alpha+\beta}(y).
\end{align}
In fact, equation \eqref{WW} follows directly by setting $\alpha=\og_1$ and $\beta=\og_2$ in \eqref{PhiAdd} which implies $\alpha+\beta=-\og_3$
and using the relation $\Phi_{-\og_3}(x)=\Phi_{\og_3}(x)\re^{-2\eta_3x}$ which follows from quasi-periodicity of the $\sigma$-function
\begin{align*}
\sigma(x+2\og_1)=-\sigma(x)\re^{2\eta_1(x+\og_1)} \quad \hbox{and} \quad \sigma(x+2\og_2)=-\sigma(x)\re^{2\eta_2(x+\og_2)}.
\end{align*}
Furthermore, we have the following addition formulae:
\bse\label{WRel}
\begin{align}
&W_1^2(x)+e_1=W_2^2(x)+e_2=W_3^2(x)+e_3=\wp(x), \label{WRela} \\
&W_1(x)W_2(x)W_3(x)=-\frac{1}{2}\wp'(x), \label{WRelb}
\end{align}
\ese
where $e_1=\wp(\og_1)$, $e_2=\wp(\og_2)$, $e_3=\wp(\og_3)$.
Equations \eqref{WRel} are essentially reformulations of some well-known addition formulae for the Weierstrass elliptic functions.
For example, equations in \eqref{WRela} follow from
\begin{align*}
\Phi_\alpha(x)\Phi_\alpha(-x)=\wp(\alpha)-\wp(x),
\end{align*}
by setting $\alpha=\og_1,\og_2,\og_3$ successively.
We refer the reader to page 397 of the monograph \cite{HJN16} for those addition formulae in terms of the $W$-functions and their proofs.

Let us now single out one of the half-periods $\og_1$, and the corresponding function $W_1(x)$, for which we have the relation
\bse\label{WP}
\begin{align}\label{WPa}
W_1(x)W_1(x+\og_1)=-\frac{\re^{\eta_1\og_1}}{\sg^2(\og_1)},
\end{align}
and subsequently
\begin{align}\label{WPb}
\left(\wp(x)-e_1\right)\left(\wp(x+\og_1)-e_1\right)=\frac{\re^{2\eta_1\og_1}}{\sg^4(\og_1)}=(e_1-e_2)(e_1-e_3)=\frac{1}{2}\wp''(\og_1),
\end{align}
\ese
We also introduce the corresponding parameters
\begin{align}\label{kk}
k=\zeta(\kp+\og_1)-\zeta(\kp)-\eta_1=-\frac{W_2(\kp)W_3(\kp)}{W_1(\kp)} \quad \hbox{and} \quad K=\wp(\kp)-e_1,
\end{align}
for which we note that the relations \eqref{WPb}, together with \eqref{WRelb} and \eqref{kk} lead to the elliptic curve in the rational form \eqref{Curve}, i.e.
\begin{align*}
k^2=K+3e+\frac{g}{K},
\end{align*}
where $e\doteq e_1=\wp(\og_1)$ and $g\doteq(e_1-e_2)(e_1-e_3)$.
For convenience below, we suppress the suffix $1$ from the corresponding functions when we single out $\og_1$,
i.e. we set $\og\doteq\og_1$, $\eta\doteq\eta_1=\zeta(\og_1)$ and $W(x)\doteq W_1(x)=\Phi_{\og_1}(x)\re^{-\eta_1x}$.

The aim is to express the elliptic Cauchy matrix $\Omega(\kp_i,\kp_j)$ in terms of the $W$-function
which, in turn, allows us to apply the famous Frobenius determinant formula \cite{Fro82} for elliptic Cauchy matrices.
In terms of the function $W$ we have the following key relation:
\begin{align}\label{Key}
\frac{W(\kp+\kp')}{W(\kp)W(\kp')}=-\frac{k-k'}{K-K'} \quad \hbox{or equivalently} \quad
\frac{W(\kp-\kp')}{W(\kp)W(\kp')}=\frac{k+k'}{K-K'},
\end{align}
and furthermore we can identify the Cauchy kernel of the elliptic KP system as follows:
\begin{align}\label{EllCauchy}
\Og(\kp,\kp')=\frac{K-K'}{k+k'}=-\frac{W(\kp)W(\kp')}{\sqrt{g}}W(\kp-\kp'+\og).
\end{align}
To compute the determinants and Pfaffians, we note that since the function $W$ is essentially a Lam\'e function $\Phi_\og$,
we can apply the Frobenius determinant formula for elliptic Cauchy matrices:
\begin{align}\label{Frob}
&\det\left(\Phi_\alpha(\kp_i-\kp'_j)\right)_{i,j=1,\cdots,N} \nonumber \\
&=\frac{\sg\left(\alpha+\sum_i(\kp_i-\kp'_i)\right)}{\sg(\alpha)}
\frac{\displaystyle\prod_{1\leq i<j\leq N}\sg(\kp_i-\kp_j)\sg(\kp'_j-\kp'_i)}{\displaystyle\prod_{1\leq i,j\leq N}\sg(\kp_i-\kp'_j)},
\end{align}
cf. \cite{Fro82}. Setting $\alpha=\og$ in \eqref{Frob} we get the determinant formula that is relevant for our case, which reads
\begin{align}\label{WFrob}
&\det\left(W\left(\kp_{i_{\nu}}-\kp_{i_{\nu'}}+\og\right)\right)_{\nu,\nu'=1,\cdots,m} \nonumber \\
&=\frac{\sg((m+1)\og)}{\sg^{m+1}(\og)}\,e^{-m\eta\og}\,
\frac{\displaystyle\prod_{1\leq\nu<\nu'\leq m}\sg(\kp_{i_{\nu}}-\kp_{i_{\nu'}})\sg(\kp_{i_{\nu'}}-\kp_{i_{\nu}})}
{\displaystyle\prod_{\substack{1\leq\nu,\nu'\leq m \\ \nu\neq\nu'}}\sg(\kp_{i_{\nu}}-\kp_{i_{\nu'}}+\og)},
\end{align}
which vanishes if $m$ is odd in accordance with the Pfaffian structure.
Subsequently, \eqref{EllCauchy} tells us that the elliptic Cauchy matrix $\left(\Omega(\kp_{i_{\nu}},\kp_{i_{\nu'}})\right)_{\nu,\nu'=1,\cdots,m}$
is essentially $\left(W\left(\kp_{i_{\nu}}-\kp_{i_{\nu'}}+\og\right)\right)_{\nu,\nu'=1,\cdots,m}$ up to diagonal factors
$-\frac{W(\kp_{i_\nu}+\og)W(\kp_{i_{\nu'}})}{\sqrt{g}}$ which can be converted to factors $K_{i_\nu}$ using \eqref{WRel} and \eqref{WP}.
Thus, we end up with
\begin{align}\label{OgFrob}
&\det\left(\Omega(\kp_{i_{\nu}},\kp_{i_{\nu'}})\right)_{\nu,\nu'=1,\cdots,m} \nonumber \\
&=(-1)^{m(m+1)/2}\frac{\sg((m+1)\og)}{\sg^{(m+1)^2}(\og)}\left(\prod_{\nu=1}^m\frac{K_{i_\nu}}{g}\right)
\prod_{1\leq\nu<\nu'\leq m}\frac{1}{W^2(\kp_{i_{\nu}}-\kp_{i_{\nu'}})}.
\end{align}
Furthermore, the prefactor satisfies
\begin{align*}
\frac{\sg((m+1)\og)}{\sg^{(m+1)^2}(\og)}=
\left\{
\begin{array}{ll}
0, & \hbox{if $m$ odd}, \\
(-1)^{m/2}\gamma_m^2, & \hbox{if $m$ even},
\end{array}
\right.
\end{align*}
where $\gamma_m= g^{m(m+2)/8}$.
This can be proven by induction using the quasi-periodicity of the $\sigma$-function and the relations \eqref{WP}, leading to $\gamma_{2n+2}/\gamma_{2n}=g^{n+1}$.
Therefore, the Pfaffian of the skew-symmetric elliptic Cauchy kernel takes the form
\begin{align}\label{CauchyPfaff}
\pf\left(\Omega(\kp_{i_{\nu}},\kp_{i_{\nu'}})\right)_{\nu,\nu'=1,\cdots,m}
=\frac{g^{m(m-2)/8}}{\left(\displaystyle\prod_{\nu=1}^m\,K_{i_\nu}\right)^{(m-2)/2}}\prod_{1\leq\nu<\nu'\leq m}\frac{K_{i_\nu}-K_{i_{\nu'}}}{k_{i_{\nu}}+k_{i_{\nu'}}},
\end{align}
for $m$ even, while for $m$ odd the Pfaffian vanishes.
In \eqref{CauchyPfaff} we have used \eqref{Key} and \eqref{WRela} to express the fully skew-symmetric product in rational form.

\begin{remark}
As a curiosity we note that, as a consequence of the Frobenius--Stickelberger determinant formula, i.e. the elliptic van der Monde-determinant (see \cite{FS77}),
the prefactor can be written as a Hankel determinant in the following form
\begin{align*}
\frac{\sg((m+1)x)}{\sg^{(m+1)^2}(x)}=\frac{(-1)^{m^2}}{(1!2!\cdots m!)^2}
\left|
\begin{array}{ccccc}
\wp'(x) & \wp''(x) & \cdots & \wp^{(m)}(x) \\
\wp''(x) & \wp'''(x) & \cdots & \wp^{(m+1)}(x) \\
\vdots & \vdots & & \vdots \\
\wp^{(m)}(x) & \wp^{(m+1)}(x)& \cdots & \wp^{(2m-1)}(x)
\end{array}
\right|,
\end{align*}
see Example 20.21 in \cite{WW21} and references therein, which vanishes at $x=\omega$ when $m$ is odd while for $m=2n$ even yields (up to a sign) a perfect square, namely
\begin{align}\label{gamma}
\frac{\sg((2n+1)\og)}{\sg^{(2n+1)^2}(\og)}=(-1)^n\gamma_m^2,
\end{align}
with
\begin{align*}
\gamma_m\doteq\frac{1}{1!2!\cdots(2n)!}
\left|
\begin{array}{ccccc}
\wp^{(2)}(\og) & \wp^{(4)}(\og) & \cdots & \wp^{(2n)}(\og) \\
\wp^{(4)}(\og) & \wp^{(6)}(\og) & \cdots & \wp^{(2n+2)}(\og) \\
\vdots & \vdots & & \vdots \\
\wp^{(2n)}(\og) & \wp^{(m+1)}(\og)& \cdots & \wp^{(4n-2)}(\og)
\end{array}
\right|,
\end{align*}
which miraculously turns out to be a pure power of the modulus $g$ alone, even though the individual entries depend on $g$ and $e$.
\end{remark}

\end{appendix}

\renewcommand{\bibname}{References}
\bibliography{References}
\bibliographystyle{plain}

\end{document}